\documentclass[a4paper,11pt]{article}
\usepackage[margin=1in,ignorehead, ignorefoot]{geometry}

\usepackage[utf8]{inputenc}

\usepackage{amsmath}
\usepackage{amsfonts}
\usepackage{amsthm}
\usepackage{amssymb}
\usepackage{mathtools}
\mathtoolsset{centercolon}
\usepackage{thmtools, thm-restate}

\usepackage{lmodern}
\usepackage{enumitem}
\usepackage{ifthen}
\usepackage{tikz}
\usepackage{url}
\usepackage{stmaryrd}
\usepackage{xparse}
\usepackage{bbm}
\usepackage{xcolor}
\usepackage{etoolbox}
\usepackage[subrefformat=parens]{subcaption}

\let\pathbibtex\path

\usetikzlibrary{decorations.pathmorphing}
\usetikzlibrary{decorations.pathreplacing}
\usetikzlibrary{calc}
\usetikzlibrary{positioning}
\usetikzlibrary{patterns}
\usetikzlibrary{patterns.meta}

\pgfdeclarelayer{background}
\pgfdeclarelayer{foreground}
\pgfsetlayers{background,main,foreground}

\usepackage[colorlinks]{hyperref}
\definecolor{lightblue}{rgb}{0.5,0.5,1.0}
\definecolor{darkred}{rgb}{0.5,0,0}
\definecolor{darkgreen}{rgb}{0,0.5,0}
\definecolor{darkblue}{rgb}{0,0,0.5}

\hypersetup{
	breaklinks=true,
	colorlinks=true,
	citecolor=darkblue,
	linkcolor=darkred,
	urlcolor=darkred,
	filecolor=darkgreen
	pdflang={en},
	pdftitle={Compressing CFI Graphs and Lower Bounds for the Weisfeiler-Leman Refinements},
	pdfauthor={Martin Grohe, Moritz Lichter, Daniel Neuen, Pascal Schweitzer}
}

\theoremstyle{plain}
\newtheorem{lemma}{Lemma}
\newtheorem{theorem}[lemma]{Theorem}
\newtheorem{corollary}[lemma]{Corollary}
\theoremstyle{definition}
\newtheorem{definition}[lemma]{Definition}
\newtheorem{remark}[lemma]{Remark}

\theoremstyle{plain}
\newtheorem{claim}{Claim}
\usepackage{chngcntr}
\counterwithin*{claim}{lemma}

\AtBeginEnvironment{proof}{\setcounter{claim}{0}}
\newenvironment{claimproof}[1][\proofname]{\begin{claimprooftemp}[#1]}{\end{claimprooftemp}}

\renewcommand{\phi}{\varphi}
\renewcommand{\epsilon}{\varepsilon}
\newcommand{\nat}{\mathbb{N}}
\newcommand{\ZZ}{\mathbb{Z}}
\newcommand{\FF}{\mathbb{F}}

\newcommand{\defining}[1]{\textbf{#1}}

\newcommand{\iso}{\cong}

\DeclarePairedDelimiter\set{\lbrace}{\rbrace}

\DeclarePairedDelimiterX\setcond[2]{\{}{\}}{\mathchoice{\,}{}{}{}#1 \;\delimsize\vert\; #2\mathchoice{\,}{}{}{}}

\newcommand{\msetopen}[0]{\{\hspace{-3pt}\{}
\newcommand{\msetclose}[0]{\}\hspace{-3pt}\}}
\newcommand{\mset}[1]{\msetopen #1 \msetclose}

\newcommand{\Bigmsetopen}[0]{\Big\{\hspace{-4pt}\Big\{}
\newcommand{\Bigmsetclose}[0]{\Big\}\hspace{-4pt}\Big\}}
\newcommand{\Bigmsetcondition}[2]{\Bigmsetopen\mathchoice{\,}{}{}{} #1 \;\Big\vert\; #2 \mathchoice{\,}{}{}{}\Bigmsetclose}

\newcommand{\auto}{\phi}
\newcommand{\autoA}{\phi}
\newcommand{\autoB}{\psi}

\newcommand{\tup}[1]{\bar{#1}}

\newcommand{\vertA}{u}
\newcommand{\vertB}{v}
\newcommand{\vertC}{w}

\DeclareMathSymbol{\shortminus}{\mathbin}{AMSa}{"39}
\newcommand{\inv}[1]{#1^{\shortminus 1}}

\newcommand{\bigO}{\mathcal{O}}

\newcommand{\coloring}{\chi}

\newcommand{\CFIsym}{\mathsf{CFI}}
\newcommand{\CFI}[1]{\CFIsym(#1)}

\newcommand{\compCFI}[2]{\CFI{#1}/_{#2}}
\newcommand{\precompCFI}[2]{{(\CFI{#1},#2)}}

\newcommand{\kwliter}[3]{#2_{#1}^{(#3)}}

\newcommand{\kequiv}[1]{\simeq_{#1}}
\newcommand{\kequivr}[2]{\simeq_{#1}^{#2}}

\newcommand{\compEquiv}{\equiv}
\newcommand{\compExtEquiv}{\compEquiv^*}
\newcommand{\pseudoclass}[1]{\eqclass{\approx}{#1}}

\newcommand{\eqclass}[2]{#2/_{#1}}

\definecolor{colA}{RGB}{0,80,250}
\definecolor{colB}{RGB}{60,200,230}
\definecolor{colC}{RGB}{60,180,75}
\definecolor{colD}{RGB}{230,25,75}
\definecolor{colE}{RGB}{245,130,48}
\definecolor{colF}{RGB}{145,30,180}
\definecolor{colG}{RGB}{240,50,230}
\definecolor{colH}{RGB}{170,110,40}

\usepackage{titling}

\title{Compressing CFI Graphs and Lower Bounds for the Weisfeiler-Leman Refinements\thanks{
	The first author is funded by the European Union (ERC, SymSim, 101054974).
	The second and fourth author are also funded by the European Union (ERC, EngageS, 820148).
	Views and opinions expressed are however those of the author(s) only and do not necessarily reflect those of the European Union or the European Research Council.
	Neither the European Union nor the granting authority can be held responsible for them.}
}

\author{Martin Grohe\\ \small RWTH Aachen University \and
	Moritz Lichter\\ \small RWTH Aachen University \and
	Daniel Neuen\\ \small University of Bremen \and
	Pascal Schweitzer\\ \small TU Darmstadt
}

\date{}

\begin{document}
\renewcommand{\abstractname}{\vspace{-\baselineskip}}

\maketitle
\vspace{-2.6em}
\begin{abstract}
	\noindent\textbf{Abstract.} The $k$-dimensional Weisfeiler-Leman ($k$-WL) algorithm is a simple combinatorial algorithm that was originally designed as a graph isomorphism heuristic.
	It naturally finds applications in Babai's quasipolynomial-time isomorphism algorithm, practical isomorphism solvers, and algebraic graph theory.
	However, it also has surprising connections to other areas such as logic, proof complexity, combinatorial optimization, and machine learning.

	The algorithm iteratively computes a coloring of the $k$-tuples of vertices of a graph.
	Since Fürer's linear lower bound [ICALP 2001], it has been an open question whether there is a super-linear lower bound for the iteration number for $k$-WL on graphs.
	We answer this question affirmatively, establishing an $\Omega(n^{k/2})$-lower bound for all $k$.
\end{abstract}

\section{Introduction}

The Weisfeiler-Leman algorithm is a simple combinatorial algorithm that
iteratively colors tuples of vertices of a graph, attempting to
assign different colors to tuples that are structurally different (i.e., they belong to
different orbits of the automorphism group). While it is
known that this goal cannot be reached on all graphs~\cite{CaiFI92},
the algorithm still serves as a powerful heuristic for computing
automorphisms and isomorphism of graphs. The 2-dimensional
``classical'' Weisfeiler-Leman algorithm~\cite{WeisfeilerL68},
coloring pairs of vertices, is rooted in algebraic graph theory and
closely linked to structures known as coherent configurations. The
$k$-dimensional version, coloring $k$\nobreakdash-tuples, was introduced
later by Babai and Mathon (see~\cite{CaiFI92}). It plays an important
role in Babai's quasipolynomial-time isomorphism algorithm~\cite{Babai16}.
 Remarkably, the Weisfeiler-Leman algorithm is not only relevant
in the context of the graph isomorphism problem, but has surprising
connections to important and seemingly unrelated concepts in logic~\cite{CaiFI92,ImmermanL90},
combinatorics~\cite{Dvorak10},
combinatorial optimization~\cite{AtseriasM13,GroheKMS14,GroheO15,Malkin14},
proof complexity~\cite{AtseriasF23,BerkholzG15},
and machine learning~\cite{MorrisLMRKGFB21,MorrisRFHLRG19,ShervashidzeSLMB11,XuHLJ19}.

The $k$-dimensional Weisfeiler-Leman algorithm ($k$-WL) colors the
$k$-tuples of vertices of a graph in a sequence of refinement
rounds. Initially, each tuple is colored by the isomorphism type of
the subgraph it induces. In each round the coloring (more
precisely, the partition it induces) is refined by comparing the color patterns
of the neighborhoods of tuples, where the neighborhood of a tuple
consist of all tuples that only differ in one position. The algorithm
stops once a stable coloring is reached, which means that the
partition into color classes is no longer refined in subsequent steps.

The two main parameters of the algorithm are the \emph{dimension} and
the \emph{iteration number}, which is the maximum number of refinement
steps it needs as function of the number of vertices. Both
parameters have been studied to great depth (see~\cite{Kiefer20}). In
this paper, we focus on the iteration number. A trivial upper bound on
the iteration number of $k$-WL is $n^k-1$, where $n$ is the number of vertices.
For $1$-WL, this
trivial upper bound $(n-1)$ can actually be
reached~\cite{KieferM20}. Maybe surprisingly, this is no longer the
case for $2$-WL \cite{KieferS19}. In fact, the iteration number of
$2$-WL is in $\bigO(n\log n)$~\cite{LichterPS19} and, more generally, the
iteration number of $k$-WL is in $\bigO(n^{k-1}\log n)$ for all $k\ge 2$~\cite{GroheLN23}.
Not much is known in terms of lower bounds, at least
not on graphs. The WL algorithm can also be applied to arbitrary
relational structures, and it was proved in~\cite{BerkholzN23} that
the iteration number of $k$-WL on $k$-ary relational structures is
$n^{\Omega(\frac{k}{\log k})}$. This was recently improved to
$n^{\Omega(k)}$~\cite{GroheLN23}. However, for $k$-WL on graphs the
best-known lower bound so far has been $\Omega(n)$~\cite{Furer01}.

\begin{theorem}
	\label{thm:main}
	The maximum iteration number of $k$-WL on graphs is in $\Omega(n^{k/2})$.
\end{theorem}

Our lower bound is not only the first super-linear lower bound for the iteration number of $k$-WL on graphs, it even improves the lower bounds that have previously been known for $k$-ary relational structures (while using only the binary edge relation of graphs).
We also remark that, similar to~\cite{BerkholzN23,Furer01,GroheLN23}, our construction actually provides pairs of graphs which are only distinguished by $k$-WL after $\Omega(n^{k/2})$ many rounds.
Finally, we note that our lower bound essentially remains valid even if~$k$ depends on~$n$.
Indeed, our construction works for all $k = \bigO(n^{1/2 - \varepsilon})$ where $\varepsilon > 0$ is some arbitrarily small constant.
However, in this case the lower bound becomes a little weaker due to some factors depending on~$k$ (see Corollary~\ref{cor:lower-bound-iteration-number} for the precise result).

\paragraph{Techniques.}
In a nutshell, our proof starts from the linear lower bound construction by Fürer~\cite{Furer01} and applies a novel compression technique to reduce the size of the graphs.
The lower bounds for $k$-ary relational structures~\cite{BerkholzN23,GroheLN23} follow a similar general strategy, but apply a different compression technique, introduced by Razborov~\cite{Razborov16} in the context of proof complexity.
Our technique not only avoids the introduction of $k$-ary constraints and hence $k$-ary relations, but also leads to a stronger compression.
We note that, after its first publication, our construction has already been used in proof complexity to show treelike-size vs.~width trade-offs for resolution~\cite{BerholzLV24}.

Let us highlight a few more details. Our lower bound construction starts with the classic CFI construction used in~\cite{Furer01} to show a linear lower bound ($\Omega(n)$) for the iteration number of~$k$\nobreakdash-WL.
This construction uses a~$k\times n$ grid as a base graph, in which each vertex is replaced by
a certain gadget. The central idea behind our construction, borrowed from another somewhat unrelated iterative process~\cite{Schweitzer13}, is to reuse gadgets.
We partition the vertices of the base graph into classes, and use for every class the same gadget instead of a new copy for each vertex in the original construction.
At first sight this might seem absurd because the whole point of the gadgets in the CFI construction is that they can function independently as long as they originate from non-adjacent base vertices.
The challenge is thus to partition the base vertices in such a manner that the iteration number remains largely unaffected, but the CFI graph shrinks.
Our initial computer experiments showed that arbitrary reuse does not have the desired effect and instead the iteration number decreases dramatically.
However, reusing gadgets for base vertices far apart showed promising behavior in the experimentally observed number of iterations.
The actual partition must therefore be chosen carefully.
Indeed, in our final theoretical construction, base vertices in each row are partitioned in a periodic manner.
However we need to ensure that reused gadgets are conceptually ``sufficiently far apart'' to not interfere with gadgets of other rows.
One way to achieve this is to ensure that periods for different rows are chosen to be mutually coprime.
However, there is a complication, namely that choosing the periods to be mutually coprime leads to a type of ``global dependency'' among the gadgets.
We fix this with a compromise of choosing the periods to be as coprime as possible while avoiding the global dependency.
(The technical requirement is that every prime-power that appears in some period must appear in at least two of the periods.)

For a theoretical analysis of the iteration number, reusing gadgets for many base vertices leads to a significant increase in complexity.
To handle this complexity, our arguments exploit a connection between cops and robber games and CFI graphs.
For our purposes, we have to extend this connection to take into account the partitions we introduced.
(In the new game rather than blocking a single vertex, a cop now blocks an entire class of the partition.)

\section{Preliminaries}

We denote by $[\ell]$ the set $\set{1,\dots,\ell}$,
by~$[0,\ell]$ the set~$\set{0,\dots,\ell}$,
and by $\mset{a_1,\dots,a_m}$ the multiset containing the elements $a_1,\dots,a_m$.
A \defining{colored graph} is a tuple $G=(V,E,c)$,
where $c \colon V \to \nat$ assigns colors to vertices.
The vertex set~$V$ of~$G$ is denoted by $V(G)$ and the edge set by $E(G)$.
We only consider simple graphs in this paper.
The distance of two sets $W,W' \subseteq V$
is the minimal distance of two vertices $\vertA \in W$ and $\vertA' \in W'$.
Isomorphisms of colored graphs have to be color-preserving,
i.e., have to map vertices of color $i$ to vertices of color $i$.
We also consider (colored) graphs with an equivalence relation on the vertices, which we denote by $(G,\equiv)$ or $(V,E,c,\equiv)$.
The equivalence relation has to be preserved by isomorphisms.

\paragraph{The Weisfeiler-Leman Algorithm.}
We describe the $k$\nobreakdash-dimensional Weisfeiler-Leman algorithm ($k$-WL).
Fix $k \geq 2$ and an arbitrary colored graph $G=(V,E,c)$.
The $k$\nobreakdash-WL algorithm computes an isomorphism-invariant coloring of $k$\nobreakdash-tuples of~$V$.
Such a coloring is a function $\coloring\colon V^k \to C$ for some finite set of colors~$C$.
For two such colorings $\coloring_i\colon V^k \to C_i$ for $i \in [2]$,
the coloring~$\coloring_1$ \defining{refines}~$\coloring_2$
if $\coloring_1(\tup{\vertA}) = \coloring_1(\tup{\vertB})$
implies $\coloring_2(\tup{\vertA}) = \coloring_2(\tup{\vertB})$
for all $\tup{\vertA},\tup{\vertB} \in V^k$.
The colorings are \defining{equivalent} if~$\coloring_1$ refines~$\coloring_2$ and~$\coloring_2$ refines~$\coloring_1$.
In this case, both colorings induce the same partition of~$V^k$ into color classes.

In the \defining{initial coloring} $\kwliter{k}{\coloring}{0}$,
two $k$\nobreakdash-tuples $\tup{\vertA} = (\vertA_1,\dots, \vertA_k),\tup{\vertB} = (\vertB_1,\dots, \vertB_k)\in V^k$ get the same color $\kwliter{k}{\coloring}{0}(\tup{\vertA})=\kwliter{k}{\coloring}{0}(\tup{\vertB})$ if and only if the map defined via $\vertA_i\mapsto\vertB_i$ for every ${i \in [k]}$ is an isomorphism from the induced subgraph $G[\set{\vertA_1,\dots, \vertA_k}]$ to the induced subgraph $G[\set{\vertB_1,\dots, \vertB_k}]$.
The computed coloring is refined iteratively via
\[\kwliter{k}{\coloring}{r+1}(\tup{\vertA}) \coloneqq \Big(\kwliter{k}{\coloring}{r}(\tup{\vertA}), \Bigmsetcondition{\bigl(\kwliter{k}{\coloring}{r}(\tup{\vertA}[\vertC/1]), \dots, \kwliter{k}{\coloring}{r}(\tup{\vertA}[\vertC/k])\bigr)}{\vertC \in V}\Big)\]
for every $\tup{\vertA} \in V^k$.
Here, $\tup{\vertA}[\vertC/i]$ denotes the tuple obtained from~$\tup{\vertA}$
by replacing the $i$-th entry with the vertex~$\vertC$.
Because $\kwliter{k}{\coloring}{r}(\tup{\vertA})$ is included in the new color of~$\tup{\vertA}$,
the coloring~$\kwliter{k}{\coloring}{r+1}$ refines~$\kwliter{k}{\coloring}{r}$.
Because we color $k$\nobreakdash-tuples,
there is an $r < |V|^k$ such that $\kwliter{k}{\coloring}{r}$ is equivalent to~$\kwliter{k}{\coloring}{r+1}$.
We say that $k$\nobreakdash-WL \defining{stabilizes} after at most~$r$ iterations.
The \defining{iteration number} of~$k$\nobreakdash-WL on~$G$ is the minimal number~$r$
such that $k$\nobreakdash-WL stabilizes after~$r$ iterations.

Let~$H$ be another colored graph with vertex set~$W$.
We now refer with $\kwliter{k}{\coloring}{r}[G]$ and $\kwliter{k}{\coloring}{r}[H]$ to the colorings
computed by $k$\nobreakdash-WL after~$r$ iterations on~$G$ and~$H$, respectively.
We say that $k$\nobreakdash-WL \defining{distinguishes the colored graphs~$G$ and~$H$}
after~$r$ rounds if there is a color~$q$ such that
\[\left|\setcond*{\tup{\vertA} \in V^k}{\kwliter{k}{\coloring}{r}[G](\tup{\vertA}) = q}\right| \neq \left|\setcond*{\tup{\vertB} \in W^k}{\kwliter{k}{\coloring}{r}[H](\tup{\vertB}) = q}\right|.\]
Note that if $k$\nobreakdash-WL distinguishes~$G$ and~$H$ after $r+1$ rounds but not after~$r$,
then the iteration number of $k$\nobreakdash-WL on~$G$ and on~$H$ is at least $r+1$.

We also apply the WL algorithm to pairs $(G,\compEquiv)$ of a graph~$G$ extended by an equivalence relation~$\compEquiv$.
The only difference there is that in the initial coloring $\kwliter{k}{\coloring}{0}$ the isomorphisms between the induced subgraphs of tuples also need to preserve the equivalence relation.

\paragraph{The Bijective Pebble Game.}
It is well-known that the distinguishing power of $k$\nobreakdash-WL is captured by the following game.
The \defining{bijective $k$-pebble game}~\cite{Hella96}
is a game played by two players called Spoiler and Duplicator
on two colored graphs~$G$ and~$H$ with vertex sets~$V$ and~$W$, respectively.
There is a pair of pebbles~$p_i$ and~$q_i$ for every~$i\in [k]$.
The pebbles~$p_i$ can be placed on~$V$ and the pebbles~$q_i$ on~$W$.
When~$\ell$ pebble pairs are placed for some~$\ell \in [k]$,
then the \defining{position} in the game is the pair~$\tup{\vertA}, \tup{\vertB}$
for~$\tup{\vertA} = (\vertA_1, \dots, \vertA_\ell) \in V^\ell$
and~$\tup{\vertB} = (\vertB_1, \dots, \vertB_\ell) \in W^\ell$
where corresponding pebbles are placed on~$\vertA_i$ and~$\vertB_i$ for every~$i \in [\ell]$.
Spoiler wins immediately if~$|V| \neq |W|$.
Otherwise, the pebbles are initially placed beside the graphs.
A round of the game is played as follows:
\begin{enumerate}
	\item Spoiler picks up a pair of pebbles $p_i$ and $q_i$.
	\item Duplicator chooses a bijection $h\colon V \to W$.
	\item Spoiler chooses a vertex $\vertA \in V$ and
	places~$p_i$ on~$\vertA$ in~$G$ and~$q_i$ on~$h(\vertA)$ in~$H$.
\end{enumerate}
Spoiler wins in~$r$ rounds,
if after at most~$r$ rounds the current position $\tup{\vertA}, \tup{\vertB}$  does not induce a partial isomorphism.
That is, the mapping $\vertA_i \mapsto \vertB_i$ for every $i \in [\ell]$
is not an isomorphism of the induced subgraphs $G[\set{\vertA_1, \dots, \vertA_\ell}]$ and $H[\set{\vertB_1, \dots, \vertB_\ell}]$.
Duplicator wins in~$r$ rounds
if Spoiler does not win in~$r$ rounds.
Spoiler wins a play if Spoiler wins after some round.
Duplicator wins if Spoiler never wins.
Spoiler (respectively, Duplicator) has a \defining{winning strategy} in~$r$ rounds if they can force a win interdependently of the moves of the other player.
We write $G \kequivr{k}{r} H$ if Duplicator has a winning strategy in $r$ rounds
and $G, \tup{\vertA} \kequivr{k}{r} H, \tup{\vertB}$ in case Duplicator
has a winning strategy when starting in position $\tup{\vertA},\tup{\vertB}$.
When considering the game without a fixed number of rounds,
we write $G \kequiv{k} H$ if Duplicator has a winning strategy.

\begin{lemma}[\cite{Hella96, CaiFI92}]
	\label{lem:k-wl-k-plus-one-bijective-game}
	For every $k\geq 2$, all colored graphs~$G$ and~$H$,
	all $\tup{\vertA}\in V(G)^k$, $\tup{\vertB} \in V(H)^k$, and $r\in \nat$,
	we have $G, \tup{\vertA} \not\kequivr{k+1}{r} H, \tup{\vertB}$
	if and only if $\kwliter{k}{\coloring}{r}[G](\tup{\vertA}) \neq \kwliter{k}{\coloring}{r}[H](\tup{\vertB})$.

	Furthermore, $G \not\kequiv{k+1} H$ if and only if $k$\nobreakdash-WL distinguishes~$G$ and~$H$.
	The required round number of Spoiler and the iteration number of $k$\nobreakdash-WL differ by at most~$k$.
\end{lemma}

We sometimes play the game on pairs $(G,\compEquiv)$ of a graph~$G$ extended by an equivalence relation~$\compEquiv$.
The only difference there is that the partial isomorphisms in the winning condition have to preserve the equivalence relation as well.
Lemma~\ref{lem:k-wl-k-plus-one-bijective-game} remains true under this extension.

\paragraph{The Cops and Robber Game.}
We use the following game later to analyze the iteration number of the WL algorithm on a certain class of graphs.
The \defining{$k$-Cops and Robber game}~\cite{SeymourT93} is played on a graph~$G$ between~$k$ many cops and a robber.
Initially, the robber is placed on some edge of~$G$ and the cops are placed beside the graph (if~$G$ has no edge, the robber loses immediately).
A round of the game is played as follows:
\begin{enumerate}
	\item \label{itm:cop-picked-up} One cop is picked up, and a destination $v\in V(G)$ for this cop is selected.\footnote{This is often described as a cop entering a helicopter and flying to $v$.
	Before deciding on a move, the robber sees the helicopter approaching and knows that the cop will be moving to $v$.}
	\item The robber moves to an edge reachable from the edge currently occupied by the robber via a path in~$G$ that does not use any vertex occupied by a cop.
	\item  The cop that was picked up in Step~\ref{itm:cop-picked-up} is placed on $v$.
\end{enumerate}
The robber is caught if cops are placed on both endpoints of the robber-occupied edge.
If the robber is caught after at most $r$~rounds,
then the cops win in~$r$ rounds.
Otherwise, the robber wins in~$r$ rounds.
The cops (respectively, the robber) have a winning strategy in~$r$ rounds
if they can force a win in~$r$ rounds independently of the action of the other player.

Note that in the original game the robber is also placed on a vertex and is caught if a cop is placed on the same vertex.
The robber loses in one version of game exactly if the robber loses in the other and in exactly the same number of rounds.
Placing the robber on edges is more suitable for our presentation.

\section{CFI Graphs and the Linear Lower Bound for $k$-WL}
\label{sec:lin:lower}

We review the CFI construction~\cite{CaiFI92} and the linear lower bound on the iteration number of $k$-WL~\cite{Furer01}.

\paragraph{The CFI Construction.}
The CFI construction starts with a so-called base graph.
An \defining{(ordered) base graph} is a connected, colored graph such that every vertex has a unique color.
Note that if every vertex has a unique color (which is a natural number), the colors induce a linear order on the vertices.
We make use of this linear order in later constructions.
Also, when defining a base graph, we only specify the linear order on the vertices and implicitly assign colors from $\{1,\dots,n\}$ to the vertices according to the given linear order. We refer to the vertices and edges of a base graph as \defining{base vertices} and \defining{base edges}, respectively.

Let $G=(V,E,c)$ be an ordered base graph
and let $f \colon E \to \FF_2$ be some function.
The \defining{CFI graph} $\CFI{G,f}$
is a colored graph and defined as follows.
The vertices of the CFI gadget for a degree~$d$ base vertex $\vertA\in V$
are the pairs $(\vertA,\tup{a})$ for all $d$-tuples $\tup{a} =(a_1, \dots, a_d) \in \FF_2^d$ with $\sum \tup{a} = a_1 + \cdots + a_d = 0$.
We say that the vertex $(\vertA,\tup{a})$ has \defining{origin}~$\vertA$. Vertices inherit the color of their origin.
The vertices of the same origin form a gadget.
Since every vertex of the base graph has a unique color, the vertices of each gadget form a color class of the CFI graph.
For two adjacent base vertices $\vertA,\vertB \in V$,
we add the following edges between the gadgets for~$\vertA$ and~$\vertB$ to $\CFI{G,f}$:
Assume that~$\vertA$ is the $i$-th neighbor of~$\vertB$
and that~$\vertB$ is the $j$-th neighbor of~$\vertA$ (according to the order on $V$).
An edge between vertices $(\vertA, \tup{a})$ and $(\vertB, \tup{b})$ is added if and only if $a_i + b_j = f(\set{\vertA,\vertB})$,
where $a_i$ is the $i$-th entry of $\tup{a}$ and $b_j$ is the $j$-th entry of $\tup{b}$.
This means that we add two complete bipartite graphs between the gadgets for~$\vertA$ and~$\vertB$.
Figure~\ref{fig:cfi-grid} shows examples of CFI graphs.

\paragraph{Isomorphisms of CFI Graphs.}
Let $g \colon E \to \FF_2$ be another function.
We say that a base edge $e \in E$ is \defining{twisted} with respect to $f$ and $g$ if $f(e) \neq g(e)$.
Twisted edges can be moved by isomorphisms.
We represent isomorphisms by sets of directed base edges.

\begin{definition}[Twisting]
	A set $T \subseteq \setcond{(\vertA,\vertB)}{\set{\vertA,\vertB} \in E}$
	is called a \defining{$G$-twisting}
	if, for every $\vertA \in V$, the set $T \cap (\set{\vertA} \times V)$
	is of even size.
	The twisting $T$
	\begin{itemize}
		\item \defining{twists} an edge $\set{\vertA,\vertB} \in E$
		if the set $T$
		contains exactly one of $(\vertA,\vertB)$ and $(\vertB,\vertA)$ and
		\item \defining{fixes} a vertex $\vertA \in V$ if $T \cap (\set{\vertA} \times V) = \emptyset$.
	\end{itemize}
	With every $G$-twisting $T$, we associate the function $g_T \colon E \to \FF_2$
	defined via $g_T(\set{\vertA,\vertB}) = 1$
	if and only if $T$ twists $\set{\vertA,\vertB}$.
\end{definition}

Twistings are one way to make typical arguments about isomorphisms between CFI graphs more explicit. These arguments all go back to \cite{CaiFI92} and have since been developed in \cite{GradelP19,Hella96,Lichter23}.
We summarize them in the following lemma.

\begin{lemma}\label{lem:cfi-iso}
	For every ordered base graph~$G$,
	every $f,g \colon E(G) \to \FF_2$,
	and every tuple~$\tup{\vertA}$ of vertices of $\CFI{G,f}$ (and thus of $\CFI{G,g}$),
	we have $(\CFI{G,f},\tup{\vertA}) \iso  (\CFI{G,g},\tup{\vertA})$ if and only if there exists a $G$-twisting~$T$
	such that $f = g + g_T$ and~$T$ fixes all vertices in~$\tup{\vertA}$.
\end{lemma}

Intuitively, the simplest isomorphism of CFI graphs is one moving the twist between two base edges incident to the same base vertex.
These isomorphisms can be composed to isomorphisms moving the twist along a path and in fact to all other isomorphisms:
For every base vertex, for an even number of incident base edges it is changed whether the edge is twisted or not (which follows from the condition $\sum \tup{a} = 0$ on the gadget vertices).
Lemma~\ref{lem:cfi-iso} implies, in particular, that two CFI graphs $\CFI{G,f}$ and $\CFI{G,g}$ are isomorphic
if and only if modulo~$2$ we have $\sum f = \sum_{e \in E} f(e) = \sum_{e \in E} g(e) = \sum g$.

\newcommand{\id}[1]{#1}

\newcommand{\drawgadget}[5] {

	\draw[dashed, thick, gray] (#3,#4) circle (0.65);

	\ifnum #1=1
	\node[#5] (#2\id{0}) at (#3,#4){};
	\else \ifnum #1=2
	\node[#5] (#2\id{11}) at (#3+0.25,#4-0.25){};
	\node[#5] (#2\id{00}) at (#3-0.25,#4+0.25){};

	\else
	\foreach \x in {0,...,1}
	\foreach \y in {0,...,1}
	\foreach \z in {0,...,1}
	{
		\ifodd \numexpr \x + \y + \z \else
		\node[#5] (#2\x\y\z) at ({#3 + 0.5*( \z - 0.3*\x) -0.2}, {#4 + 0.5*(\x + 0.7*\z - 1.7*\y -\x*\z)+0.1}){};
		\fi
	}
	\fi
	\fi

}

\newcounter{memcount}
\newcommand{\countmem}[2]{%
	\setcounter{memcount}{0}%
	\foreach \i in #1{%
		\ifthenelse{\equal{\i}{#2}} {%
			\stepcounter{memcount}%
		}{}%
	}%

}

\newcommand{\suffx}{x}
\newcommand{\suffy}{y}
\newcommand{\suffz}{z}
\newcommand{\suffsum}{sum}
\newcommand{\suffseq}{seq}
\newcommand{\suffrr}{rr}

\newcommand{\foreachdim}[4] {
	\ifnum #1=1
	\expandafter\foreach \csname#2\suffx\endcsname in {0,...,1} {
		\expandafter\def\csname#2\suffsum\endcsname{%
			\csname#2\suffx\endcsname}
		\expandafter\def\csname#2\suffseq\endcsname{%
			\csname#2\suffx\endcsname}
		\expandafter\def\csname#2\suffrr\endcsname{%
			{\csname#2\suffx\endcsname}}
		#3
	}
	\else \ifnum #1=2
	\expandafter\foreach \csname#2\suffx\endcsname in {0,...,1} {
		\expandafter\foreach \csname#2\suffy\endcsname in {0,...,1} {
			\expandafter\def\csname#2\suffsum\endcsname{%
				\csname#2\suffx\endcsname + %
				\csname#2\suffy\endcsname}
			\expandafter\def\csname#2\suffseq\endcsname{%
				\csname#2\suffx\endcsname\csname#2\suffy\endcsname}
			\expandafter\def\csname#2\suffrr\endcsname{%
				{\csname#2\suffx\endcsname,\csname#2\suffy\endcsname}}
			#3
		}
	}
	\else
	\expandafter\foreach \csname#2\suffx\endcsname in {0,...,1} {
		\expandafter\foreach \csname#2\suffy\endcsname in {0,...,1} {
			\expandafter\foreach \csname#2\suffz\endcsname in {0,...,1} {
				\expandafter\def\csname#2\suffsum\endcsname{%
					\csname#2\suffx\endcsname + %
					\csname#2\suffy\endcsname +%
					\csname#2\suffz\endcsname}
				\expandafter\def\csname#2\suffseq\endcsname{%
					\csname#2\suffx\endcsname\csname#2\suffy\endcsname\csname#2\suffz\endcsname}
				\expandafter\def\csname#2\suffrr\endcsname{%
					{\csname#2\suffx\endcsname,\csname#2\suffy\endcsname,\csname#2\suffz\endcsname}}
				#3
			}
		}
	}
	\fi\fi

}

\newcommand{\drawconnectgadget}[9] {

	\foreachdim{#1}{a} {
		\foreachdim{#2}{b}{
			\pgfmathsetmacro\as{\arr[#5]}
			\pgfmathsetmacro\bs{\brr[#6]}
			\ifodd \numexpr \asum \else
			\ifodd \numexpr \bsum \else
			\ifthenelse{
			 \(\(\NOT \equal{#9}{twist}\) \AND \equal{\as}{\bs}\) \OR %
					\(\equal{#9}{twist} \AND \(\NOT \equal{\as}{\bs}\)\)
			}{
			\countmem{{#7}}{\aseq\id{E}\bseq}
			\ifthenelse{\value{memcount} > 0}{
				\path (#3\aseq) edge[bend right=5] (#4\bseq);
			}{
				\countmem{{#8}}{\aseq\id{E}\bseq}
				\ifthenelse{\value{memcount} > 0}{
					\path (#3\aseq) edge[bend left=5] (#4\bseq);
				}{
					\path (#3\aseq) edge (#4\bseq);
				}
			}
			}{}
			\fi\fi
		}
	}

}

\newcommand{\gadgetlabel}[5] {
	\node[fill = none] () at ($(#1#2#3#4) + #5$) {\scriptsize$#1(#2,#3,#4)$};
}

\tikzstyle{vertex} = [circle, fill=black, inner sep=0mm, minimum size = 2mm]

\tikzstyle{basevertex} = [circle, fill=black, inner sep=0mm, minimum size = 3mm]

\colorlet{gcolA}{colH}
\colorlet{gcolB}{colB}
\colorlet{gcolC}{colC}
\colorlet{gcolD}{colD}
\colorlet{gcolE}{colF}
\colorlet{gcolF}{colE}
\colorlet{gcolG}{colG}
\colorlet{gcolH}{colA}

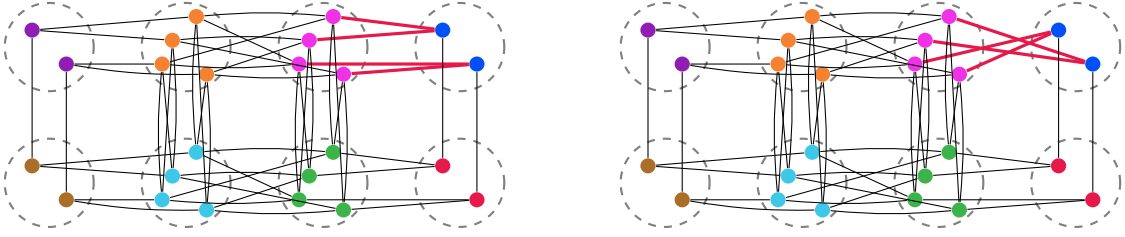
\begin{figure}
	\centering
	\begin{tikzpicture}[>=stealth,on grid,every node/.style={vertex}, scale=0.75]

		\def\spacing{2}

		\begin{scope}[shift={(-4.5,0)}]
		\drawgadget{2}{z1}{-\spacing}{0}{gcolA}
		\drawgadget{2}{z2}{-\spacing}{\spacing}{gcolE}

		\drawgadget{3}{a1}{0}{0}{gcolB}
		\drawgadget{3}{b1}{\spacing}{0}{gcolC}
		\drawgadget{2}{c1}{2*\spacing}{0}{gcolD}
		\drawgadget{3}{a2}{0}{\spacing}{gcolF}
		\drawgadget{3}{b2}{\spacing}{\spacing}{gcolG}
		\drawgadget{2}{c2}{2*\spacing}{\spacing}{gcolH}

		\drawconnectgadget{3}{3}{b1}{a1}{0}{0}{101E101,000E000}{011E011, 110E110}{}
		\drawconnectgadget{3}{2}{b1}{c1}{1}{1}{000E000,011E011}{101E101, 110E110}{}
		\drawconnectgadget{3}{3}{b2}{a2}{0}{0}{101E101,000E000}{110E110,011E011}{}
		\begin{scope}[every edge/.style={draw, very thick, gcolD}]
		\drawconnectgadget{3}{2}{b2}{c2}{1}{1}{}{}{}
		\end{scope}

		\drawconnectgadget{3}{3}{a1}{a2}{2}{2}{000E000,011E011}{101E101,110E110}{}
		\drawconnectgadget{3}{3}{b1}{b2}{2}{2}{000E000,011E011}{101E101,110E110}{}
		\drawconnectgadget{2}{2}{c1}{c2}{1}{0}{011E11}{110E00}{}

		\drawconnectgadget{3}{2}{a1}{z1}{1}{0}{}{011E11}{}
		\drawconnectgadget{3}{2}{a2}{z2}{1}{0}{}{011E11}{}
		\drawconnectgadget{2}{2}{z1}{z2}{1}{1}{}{}{}
		\end{scope}

		\begin{scope}[shift={(4.5,0)}]
			\drawgadget{2}{z1}{-\spacing}{0}{gcolA}
			\drawgadget{2}{z2}{-\spacing}{\spacing}{gcolE}

			\drawgadget{3}{a1}{0}{0}{gcolB}
			\drawgadget{3}{b1}{\spacing}{0}{gcolC}
			\drawgadget{2}{c1}{2*\spacing}{0}{gcolD}
			\drawgadget{3}{a2}{0}{\spacing}{gcolF}
			\drawgadget{3}{b2}{\spacing}{\spacing}{gcolG}
			\drawgadget{2}{c2}{2*\spacing}{\spacing}{gcolH}

			\drawconnectgadget{3}{3}{b1}{a1}{0}{0}{101E101,000E000}{011E011, 110E110}{}
			\drawconnectgadget{3}{2}{b1}{c1}{1}{1}{000E000,011E011}{101E101, 110E110}{}
			\drawconnectgadget{3}{3}{b2}{a2}{0}{0}{101E101,000E000}{110E110,011E011}{}
			\begin{scope}[every edge/.style={draw, very thick, gcolD}]
			\drawconnectgadget{3}{2}{b2}{c2}{1}{1}{}{}{twist}
			\end{scope}

			\drawconnectgadget{3}{3}{a1}{a2}{2}{2}{000E000,011E011}{101E101,110E110}{}
			\drawconnectgadget{3}{3}{b1}{b2}{2}{2}{000E000,011E011}{101E101,110E110}{}
			\drawconnectgadget{2}{2}{c1}{c2}{1}{0}{011E11}{110E00}{}

			\drawconnectgadget{3}{2}{a1}{z1}{1}{0}{}{011E11}{}
			\drawconnectgadget{3}{2}{a2}{z2}{1}{0}{}{011E11}{}
			\drawconnectgadget{2}{2}{z1}{z2}{1}{1}{}{}{}
		\end{scope}

	\end{tikzpicture}
	\caption{Two CFI graphs with the $2\times4$ grid as base graph.
	The connection highlighted in red between the blue and pink gadget is twisted between the two CFI graphs.}
	\label{fig:cfi-grid}
\end{figure}

\paragraph{CFI Graphs and the Cops and Robber Game.}

There is a close connection between the bijective $k$\nobreakdash-pebble game
played on non-isomorphic CFI graphs and the $k$\nobreakdash-Cops and Robber game played on the base graph.
Intuitively, Spoiler has to catch the twist with the pebbles.
Placing a pebble in the bijective $k$\nobreakdash-pebble game corresponds to placing a cop in the Cops and Robber game.
Moving the twisted edge corresponds to moving the robber.
In the bijective $k$\nobreakdash-pebble game, the twistings have to fix the origins of the pebbled vertices.
Analogously, the robber is not allowed to use vertices occupied by cops.
The path the robber takes is exactly the path along which the twist is moved.

This leads to the following two lemmas. A proof of the first lemma is implicit in \cite[Theorem~3]{DawarR07}.

\begin{lemma}
	\label{lem:cops-and-robbers-if-bijective}
	Let~$k\geq2$,~$r\in \nat$,~$G$ be an ordered base graph, and~$f,g \colon E(G) \to\FF_2$.
	If the robber has a winning strategy in the~$r$\nobreakdash-round~$k$\nobreakdash-Cops and Robber game played on~$G$,
	then Duplicator has a winning strategy in the~$r$\nobreakdash-round bijective~$k$\nobreakdash-pebble game played on~$\CFI{G,f}$ and~$\CFI{G,g}$.
\end{lemma}

The second lemma follows from combining \cite{SeymourT93}, \cite[Theorem 3.13]{Roberson22}, and \cite[Theorem~7]{Dvorak10} (the lemma can also be proved more directly following the above intuition, but we are not aware of an explicit reference).

\begin{lemma}
	\label{lem:cops-and-robbers-then-bijective}
	Let~$k\geq2$,~$G$ be an ordered base graph, and~$f,g \colon E(G) \to\FF_2$ such that~$f$ and~$g$ twist an odd number of edges.
	If the cops have a winning strategy in the~$k$\nobreakdash-Cops and Robber game played on~$G$,
	then Spoiler has a winning strategy in the bijective~$k$\nobreakdash-pebble game played on~$\CFI{G,f}$ and~$\CFI{G,g}$.
\end{lemma}

\paragraph{The Linear Lower Bound.}

We recall the linear lower bound on the iteration number of $k$\nobreakdash-WL by Fürer~\cite{Furer01}.
Fix some $k \geq 2$ in this paragraph.
Let $n \geq k$ and define $I \coloneqq [0,k-1]$ and $J \coloneqq [0,n-1]$.
Let $G_{I,J}$ denote the $k \times n$ grid graph with row indices in $I$ and column indices in $J$ (i.e.,~$k$ is the height and~$n$ is the width of the grid).
Formally, $V(G_{I,J}) \coloneqq I \times J$ and
\begin{align*}
	E(G_{I,J}) \coloneqq\quad &\setcond[\big]{\{(i,j),(i,j+1)\}}{i \in I, j \in [0,n-2]}\\
	\cup~&\setcond[\big]{\{(i,j),(i+1,j)\}}{i \in [0,k-2], j \in J}.
\end{align*}
We remark that we choose indices starting from $0$ because this turns out to be more convenient for later constructions.

First note that every vertex in $G_{I,J}$ has degree at most~$4$.
Thus, every CFI gadget contains at most~$8$ vertices and hence $\CFI{G_{I,J}, f}$ has $\Theta(n)$ many vertices (for fixed~$k$).
We argue that $k$\nobreakdash-WL needs $\Omega(n)$ iterations to distinguish non-isomorphic CFI graphs $\CFI{G_{I,J}, f}$ and $\CFI{G_{I,J}, g}$ for $f,g\colon E(G_{I,J}) \to \FF_2$ with $\sum f \neq \sum g$ (cf.~Figure~\ref{fig:cfi-grid} for $k=2$ and $n=4$).
To do so, we can play the $(k+1)$\nobreakdash-Cops and Robber game by Lemmas~\ref{lem:k-wl-k-plus-one-bijective-game},~\ref{lem:cops-and-robbers-if-bijective}, and~\ref{lem:cops-and-robbers-then-bijective}.

The optimal strategy of the cops is to initially form a separator with $k$ cops that cuts the grid into two halves.
One of the two halves contains the robber.
Using the one additional cop, the separator can slowly be moved to the border of the grid (by at most $\bigO(1)$ columns per round).
Thus, after $\bigO(n)$ rounds, the robber is caught.
To see that indeed $\Omega(n)$ many rounds are needed, we see that the cops cannot move the separator faster because there is only one additional cop.
If the left border of the grid is not separated from the right border, then the robber can move between the two borders.
If the cops form a separator again, the robber moves to the larger part of the grid.
Because the cops can catch the robber, we obtain that $k$-WL distinguishes $\CFI{G_{I,J}, f}$ and $\CFI{G_{I,J}, g}$, but only after $\Omega(n)$ iterations.

\section{Compressed CFI Graphs}
\label{sec:comp-cfi}

In this section, we develop techniques to reduce the size of CFI graphs.

\subsection{Basic Construction}

Our construction yields compressed CFI graphs.
We describe isomorphisms of these compressed CFI graphs by twistings again.
Fix an arbitrary ordered base graph~$G$ throughout this subsection.

\begin{definition}[Compression]
	An equivalence relation~$\compEquiv$ on $V(G)$ is a \defining{$G$-compression} if for all $\vertA,\vertA',\vertB, \vertB' \in V(G)$ it satisfies the following two conditions:
	\begin{enumerate}
		\item If $\vertA \compEquiv \vertB$, then
		$\vertA$ and $\vertB$ are non-adjacent and of the same degree.
		\item If $\set{\vertA,\vertB},\set{\vertA',\vertB'} \in E(G)$, $\vertA \compEquiv \vertA'$, $\vertB \compEquiv \vertB'$, and $\vertB$ is the $i$-th neighbor of $\vertA$ (according to the order on~$G$), then $\vertB'$ is the $i$-th neighbor of $\vertA'$.
	\end{enumerate}

\end{definition}

Let ${\compEquiv} \subseteq V^2$ be a $G$-compression.
It induces an equivalence relation on $\CFI{G,f}$ (independently of $f \colon E(G) \to \FF_2$), which we also denote by $\compEquiv$, as follows:
$(\vertA, \tup{a}) \compEquiv (\vertB, \tup{b})$
if and only if $\vertA \compEquiv \vertB$ and $\tup{a} = \tup{b}$.
We denote by $\compCFI{G,f}{\compEquiv}$ the colored graph obtained from
$\CFI{G,f}$ by contracting all $\compEquiv$\nobreakdash-equivalence classes into a single vertex.
Formally, the vertices of $\compCFI{G,f}{\compEquiv}$ are the $\compEquiv$\nobreakdash-equivalence classes $\vertA/_{\compEquiv} \coloneqq \setcond{\vertC \in V(\CFI{G,f})}{\vertC \compEquiv\vertA}$,
and there is an edge between $\vertA/_\compEquiv$ and $\vertB/_\compEquiv$
if there are $\vertA'\compEquiv \vertA$ and $\vertB'\compEquiv \vertB$ such that there is an edge between~$\vertA'$ and~$\vertB'$ in $\CFI{G,f}$.
Observe that $\compCFI{G,f}{\compEquiv}$ is loop-free by our condition on $\compEquiv$ that equivalent vertices of~$G$ are non-adjacent.
The color of a $\compEquiv$\nobreakdash-equivalence class in $\compCFI{G,f}{\compEquiv}$ is the minimal color of one of its members in $\CFI{G,f}$.
We need to require that~$f$ is compatible with~$\compEquiv$ in the following sense to obtain reasonable graphs.

\begin{definition}
	A function $f\colon E(G) \to \FF_2$ is \defining{$\compEquiv$-compressible}
	if, for all $\vertA,\vertB,\vertA',\vertB' \in V$,
	we have that if $\set{\vertA,\vertB},
	\set{\vertA',\vertB'} \in E(G)$,
	$\vertA \compEquiv \vertA'$, and
	$\vertB \compEquiv\vertB'$,
	then $f(\set{\vertA,\vertB}) = f(\set{\vertA',\vertB'})$.
\end{definition}

\begin{definition}[Compressed CFI]
	For a $G$-compression~$\compEquiv$ and a $\compEquiv$-compressible ${f \colon E(G) \to \FF_2}$
	\begin{itemize}
		\item the graph $\precompCFI{G,f}{\compEquiv}$ obtained from extending the colored graph $\CFI{G,f}$ with the relation~$\compEquiv$ is a \defining{precompressed CFI graph} and
		\item the colored graph $\compCFI{G,f}{\compEquiv}$ is a \defining{compressed CFI graph}.
	\end{itemize}
\end{definition}

An example is given in Figure~\ref{fig:compressed-cfi-graphs}.

\begin{lemma}
	\label{lem:compressed-cfi-graph-size}
	Let~$\compEquiv$ be a $G$\nobreakdash-compression
	and $f \colon E(G) \to \FF_2$ be $\compEquiv$\nobreakdash-compressible.
	If~$G$ is of maximum degree $d$
	and there are~$n$ many $\compEquiv$\nobreakdash-equivalence classes,
	then $\compCFI{G,f}{\compEquiv}$
	has at most $2^{d-1} n$ vertices.
\end{lemma}

\begin{proof}
	After contracting~$\compEquiv$ on $\precompCFI{G,f}{\compEquiv}$,
	there are exactly~$n$ gadgets in $\precompCFI{G,f}{\compEquiv}$ whose vertices
	remain in $\compCFI{G,f}{\compEquiv}$.
	For every base vertex of degree~$d'$, its gadget contains exactly~$2^{d'-1}$ vertices.
	Because every base vertex has degree at most~$d$,
	the graph $\compCFI{G,f}{\compEquiv}$ has at most $2^{d-1} n$ vertices.
\end{proof}

\begin{lemma}
	\label{lem:edges-compressed-precompressed}
	Let~$\compEquiv$ be a $G$\nobreakdash-compression
	and let $f \colon E(G) \to \FF_2$ be $\compEquiv$\nobreakdash-compressible.
	For all $\vertA, \vertB \in V(\CFI{G,f})$,
	the set $\set{\vertA/_\compEquiv,\vertB/_\compEquiv}$ is an edge (or non-edge) in $\compCFI{G,f}{\compEquiv}$
	if and only if, for all (or equivalently some) $\vertA' \compEquiv \vertA$ and $\vertB' \compEquiv \vertB$ such that their origins are adjacent in $G$,
	the set $\set{\vertA',\vertB'}$ is an edge (or non-edge, respectively) in $\precompCFI{G,f}{\compEquiv}$.

	In particular, for all $\vertA, \vertB \in V(\CFI{G,f})$ whose origins are adjacent in~$G$,
	the set $\set{\vertA/_\compEquiv,\vertB/_\compEquiv}$ is an edge in $\compCFI{G,f}{\compEquiv}$
	if and only if $\set{\vertA,\vertB}$ is an edge
	in  $\precompCFI{G,f}{\compEquiv}$.
\end{lemma}
\begin{proof}

	Let~$\vertA$ and~$\vertB$ be vertices of $\CFI{G,f}$,
	whose origins are adjacent in~$G$.
	We first claim that $\set{\vertA,\vertB}$ is an edge (or non-edge) in $\precompCFI{G,f}{\compEquiv}$
	if and only if, for all $\vertA' \compEquiv \vertA$ and $\vertB' \compEquiv \vertB$ such that their origins are adjacent in $G$,
	the set $\set{\vertA',\vertB'}$ is an edge (or non-edge, respectively) in $\precompCFI{G,f}{\compEquiv}$.
	Let $x$ and $y$ be the origins of $\vertA$ and $\vertB$, respectively,
	and $\vertA= (x,\tup{a})$ and $\vertB = (y, \tup{b})$
	for tuples $\tup{a}$ and $\tup{b}$ over $\FF_2$ whose length is the degree of $x$ and $y$, respectively.
	Assume that $\vertA' \compEquiv \vertA$ and $\vertB' \compEquiv \vertB$
	and the origins $x'$ and $y'$ of $\vertA'$ and $\vertB'$, respectively, are adjacent in $G$.
	By the extension of the compression on the vertices of the CFI graphs,
	we have $\vertA' = (x', \tup{a})$ and $\vertB' =  (y', \tup{b})$.
	Let $\vertB$ be the $i$-th neighbor of~$\vertA$ and~$\vertA$ be the $j$-th neighbor of $\vertB$.
	By the properties of a $G$-compression,~$\vertB'$ is the $i$-th neighbor of~$\vertA'$
	and~$\vertA'$ is the $j$-th neighbor of~$\vertB'$.
	By the definition of CFI graphs and because $\set{x,y},\set{x',y'} \in E(G)$,
	the set $\set{\vertA,\vertB}$ is and edge in $\CFI{G,f}$ if and only if $a_i = b_j$ if and only if $\set{\vertA',\vertB'}$ is and edge in $\CFI{G,f}$.
	The claim follows because $\precompCFI{G,f}{\compEquiv}$ and $\CFI{G,f}$ have the same edges.

	We now prove the assertion of the lemma.
	By construction of the compressed CFI graphs,
	$\set{\vertA/_\compEquiv,\vertB/_\compEquiv}$ is an edge in $\compCFI{G,f}{\compEquiv}$
	if and only if
	there are  $\vertA' \compEquiv \vertA$ and $\vertB' \compEquiv \vertB$
	such that $\set{\vertA',\vertB'}$ is an edge in $\precompCFI{G,f}{\compEquiv}$.
	By the former claim,
	this is the case if and only if,
	for all $\vertA' \compEquiv \vertA$ and $\vertB' \compEquiv \vertB$ such that their origins are adjacent in $G$,
	the set $\set{\vertA',\vertB'}$ is an edge in $\precompCFI{G,f}{\compEquiv}$.
	For the non-edge case,
	$\set{\vertA/_\compEquiv,\vertB/_\compEquiv}$ is not an edge in $\compCFI{G,f}{\compEquiv}$
	if and only if
	for all  $\vertA' \compEquiv \vertA$ and $\vertB' \compEquiv \vertB$,
	the set $\set{\vertA',\vertB'}$ is not edge in $\precompCFI{G,f}{\compEquiv}$.
	This is the case if and only if the same statement is true
	only for $\vertA' \compEquiv \vertA$ and $\vertB' \compEquiv \vertB$
	whose origins are adjacent in~$G$,
	because vertices whose origins are not adjacent in~$G$ are never adjacent in $\CFI{G,f}$.
\end{proof}

For two compressed CFI graphs $\compCFI{G,f}{\compEquiv}$ and $\compCFI{G,g}{\compEquiv}$, it is not obvious under which conditions on~$f$,~$g$, and~$\compEquiv$ they are isomorphic.
In particular, the criterion of Lemma~\ref{lem:cfi-iso} does not extend to the compressed CFI graphs.
To address this, we consider a restricted form of twistings, which respect isomorphisms of compressed CFI graphs.

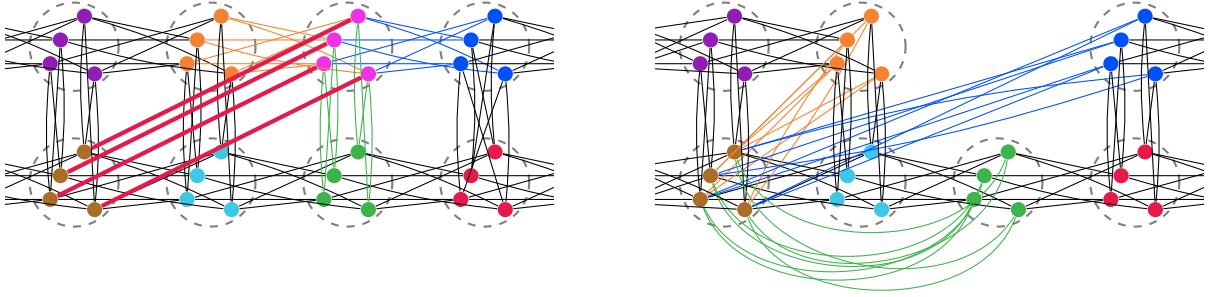
\begin{figure}
	\centering
	\begin{tikzpicture}[>=stealth,on grid,every node/.style={vertex}, scale=0.75]

		\def\spacing{2}

		\begin{scope}

			\clip (-0.5*\spacing, -0.5*\spacing) rectangle (3.5*\spacing, 1.5*\spacing);

			\drawgadget{3}{z1}{-\spacing}{0}{black}
			\drawgadget{3}{z2}{-\spacing}{\spacing}{black}

			\drawgadget{3}{a1}{0}{0}{gcolA}
			\drawgadget{3}{b1}{\spacing}{0}{gcolB}
			\drawgadget{3}{c1}{2*\spacing}{0}{gcolC}
			\drawgadget{3}{d1}{3*\spacing}{0}{gcolD}
			\drawgadget{3}{e1}{4*\spacing}{0}{black}
			\drawgadget{3}{a2}{0}{\spacing}{gcolE}
			\drawgadget{3}{b2}{\spacing}{\spacing}{gcolF}
			\drawgadget{3}{c2}{2*\spacing}{\spacing}{gcolG}
			\drawgadget{3}{d2}{3*\spacing}{\spacing}{gcolH}
			\drawgadget{3}{e2}{4*\spacing}{\spacing}{black}

			\drawconnectgadget{3}{3}{a1}{z1}{0}{1}{011E101}{011E000}{}
			\drawconnectgadget{3}{3}{a2}{z2}{0}{1}{011E101}{011E000}{}
			\drawconnectgadget{3}{3}{b1}{a1}{0}{1}{011E101}{011E000}{}
			\drawconnectgadget{3}{3}{c1}{b1}{0}{1}{011E101}{011E000}{}
			\drawconnectgadget{3}{3}{b2}{a2}{0}{1}{011E101}{011E000}{}
			\begin{scope}[gcolF]
			\drawconnectgadget{3}{3}{c2}{b2}{0}{1}{011E101}{011E000}{}
			\end{scope}
			\begin{scope}[gcolH]
			\drawconnectgadget{3}{3}{d2}{c2}{0}{1}{011E101}{011E000}{}
			\end{scope}
			\drawconnectgadget{3}{3}{d1}{c1}{0}{1}{011E101}{011E000}{}
			\drawconnectgadget{3}{3}{e2}{d2}{0}{1}{011E101}{011E000}{}
			\drawconnectgadget{3}{3}{e1}{d1}{0}{1}{011E101}{011E000}{}

			\drawconnectgadget{3}{3}{a1}{a2}{2}{2}{000E000,011E011}{101E101,110E110}{}
			\drawconnectgadget{3}{3}{b1}{b2}{2}{2}{000E000,011E011}{101E101,110E110}{}
			\begin{scope}[gcolC]
			\drawconnectgadget{3}{3}{c1}{c2}{2}{2}{000E000,011E011}{101E101,110E110}{}
			\end{scope}
			\drawconnectgadget{3}{3}{d1}{d2}{2}{2}{000E000,011E011}{101E101,110E110}{}

			\path[draw, ultra thick, gcolD]
				(a1000) to (c2000)
				(a1011) to (c2011)
				(a1110) to (c2110)
				(a1101) to (c2101);

		\end{scope}

		\begin{scope}[shift={(9.5,0)}]
			\def\spacing{2}

			\clip (-0.5*\spacing, -0.5*\spacing-0.7) rectangle (3.5*\spacing, 1.5*\spacing);

			\drawgadget{3}{z1}{-\spacing}{0}{black}
			\drawgadget{3}{z2}{-\spacing}{\spacing}{black}

			\drawgadget{3}{a1}{0}{0}{gcolA}
			\drawgadget{3}{b1}{\spacing}{0}{gcolB}
			\drawgadget{3}{c1}{2*\spacing}{0}{gcolC}
			\drawgadget{3}{d1}{3*\spacing}{0}{gcolD}
			\drawgadget{3}{e1}{4*\spacing}{0}{black}
			\drawgadget{3}{a2}{0}{\spacing}{gcolE}
			\drawgadget{3}{b2}{\spacing}{\spacing}{gcolF}
			\drawgadget{3}{d2}{3*\spacing}{\spacing}{gcolH}
			\drawgadget{3}{e2}{4*\spacing}{\spacing}{black}

			\drawconnectgadget{3}{3}{b1}{a1}{0}{1}{011E101}{011E000}{}
			\drawconnectgadget{3}{3}{c1}{b1}{0}{1}{011E101}{011E000}{}
			\drawconnectgadget{3}{3}{b2}{a2}{0}{1}{011E101}{011E000}{}
			\begin{scope}[gcolF]
			\drawconnectgadget{3}{3}{a1}{b2}{0}{1}{011E101}{011E000}{}
			\end{scope}
			\begin{scope}[gcolH]
				\drawconnectgadget{3}{3}{d2}{a1}{0}{1}{011E101}{011E000}{}
			\end{scope}
			\drawconnectgadget{3}{3}{d1}{c1}{0}{1}{011E101}{011E000}{}
			\drawconnectgadget{3}{3}{e2}{d2}{0}{1}{011E101}{011E000}{}
			\drawconnectgadget{3}{3}{e1}{d1}{0}{1}{011E101}{011E000}{}

			\drawconnectgadget{3}{3}{a1}{a2}{2}{2}{000E000,011E011}{101E101,110E110}{}
			\drawconnectgadget{3}{3}{b1}{b2}{2}{2}{000E000,011E011}{101E101,110E110}{}
			\begin{scope}[every edge/.style={draw, bend left = 70}, gcolC]
			\drawconnectgadget{3}{3}{c1}{a1}{2}{2}{}{}{}
			\end{scope}
			\drawconnectgadget{3}{3}{d1}{d2}{2}{2}{000E000,011E011}{101E101,110E110}{}

			\drawconnectgadget{3}{3}{a1}{z1}{1}{0}{011E101}{011E000}{}
			\drawconnectgadget{3}{4}{a2}{z2}{1}{0}{011E101}{011E000}{}

		\end{scope}

	\end{tikzpicture}
	\caption{On the left, a precompressed CFI graph in which the vertices of two gadgets are identified by the compression (drawn by red lines).
		On the right, the obtained compressed CFI graph where identified vertices are contracted.
		Some edges are colored for the sake of visual distinguishability.
		}
	\label{fig:compressed-cfi-graphs}
\end{figure}

\begin{definition}[Compressible Twisting]
	For a $G$-compression~$\compEquiv$,
	a $G$-twisting~$T$ is called \linebreak \defining{$\compEquiv$\nobreakdash-compressible} if the following holds for all $\vertA,\vertA' \in V$ with $\vertA \compEquiv \vertA'$:
	Let~$\vertA$ and~$\vertA'$ be of degree~$d$.
	Then for every $i\in[d]$, we have
	$(\vertA,\vertB_i) \in T$ if and only if $(\vertA',\vertB'_i) \in T$,
	where~$\vertB_i$ is the $i$-th neighbor of~$\vertA$ and~$\vertB'_i$ is the $i$-th neighbor of~$\vertA'$ (according to the order on~$G$).
\end{definition}

\begin{lemma}
	\label{lem:twistsings-equiv-parameters}
	For every $G$-compression $\compEquiv$,
	every set of base vertices $W \subseteq V$,
	all $\compEquiv$-compressible $f,g\colon E(G) \to \FF_2$, and
	every $\ell$-tuple $\tup{\vertA}$ of vertices in $\CFI{G,f}$ (and hence of $\CFI{G,g}$) such that the set of origins of all vertices in $\tup{\vertA}$ is~$W$,
	the following are equivalent:
	\begin{enumerate}
		\item $(\compCFI{G,f}{\compEquiv},\tup{\vertA}/_\compEquiv) \iso (\compCFI{G,g}{\compEquiv},\tup{\vertA}/_\compEquiv)$,
		where $\tup{\vertA}/_\compEquiv$ denotes the $\ell$\nobreakdash-tuple
		whose $i$\nobreakdash-th entry is the $\compEquiv$-equivalence class~$\vertA_i/_\compEquiv$ of the $i$-th entry~$\vertA_i$ of~$\tup{\vertA}$ for every $i \in [\ell]$,
		\item $\precompCFI{G,f}{\compEquiv,\tup{\vertA}} \iso \precompCFI{G,g}{\compEquiv, \tup{\vertA}}$, and
		\item there exists a $\compEquiv$-compressible $G$\nobreakdash-twisting
		such that $f = g + g_T$ and~$T$ fixes all vertices in~$W$.
	\end{enumerate}
\end{lemma}

\begin{proof}
	Let~$\compEquiv$ be a $G$-compression, $W\subseteq V$,
	$f,g\colon E(G) \to \FF_2$ be $\compEquiv$\nobreakdash-compressible,
	and~$\tup{\vertA}$ be an $\ell$\nobreakdash-tuple of vertices of $\CFI{G,f}$
	such that the set of origins of all vertices in~$\tup{\vertA}$ is~$W$.
	We show (1) $\Rightarrow$ (2) $\Rightarrow$ (3) $\Rightarrow$ (1).

	Assume there is an isomorphism
        $\auto\colon \compCFI{G,f}{\compEquiv} \to
        \compCFI{G,g}{\compEquiv}$ such that
        $\auto(\tup{\vertA}/_\compEquiv) = \tup{\vertA}/_\compEquiv$.
        We construct an isomorphism
        $\autoB \colon \precompCFI{G,f}{\compEquiv} \to
        \precompCFI{G,g}{\compEquiv}$.  Let~$\vertC$ be an arbitrary
        vertex of $\precompCFI{G,f}{\compEquiv}$
        and~$\vertC/_\compEquiv$ its
        $\compEquiv$\nobreakdash-equivalence class, that itself is a vertex
        in $\compCFI{G,f}{\compEquiv}$.  The equivalence
        classes~$\vertC/_\compEquiv$ and~$\autoA(\vertC/_\compEquiv)$
        have the same color~$c$.
        Let $\vertB$ be the unique base vertex of color~$c$.
        By definition of the compressed CFI graphs,
        $\vertC/_\compEquiv$ and~$\autoA(\vertC/_\compEquiv)$
        contain exactly one vertex with origin~$\vertB'$
        for every base vertex $\vertB' \compEquiv \vertB$.
        In particular,~$\autoA(\vertC/_\compEquiv)$ contains a
        vertex~$\vertC'$ of the same origin as~$\vertC$.  By
        construction, no $\compEquiv$\nobreakdash-equivalence class
        contains two vertices of the same origin.  Thus,~$\vertC'$ is
        unique.  We define $\autoB(\vertC) \coloneqq \vertC'$.  Again by
        construction,~$\autoB$ is color preserving.
        Let~$\vertB$
        and~$\vertC$ be vertices of $\CFI{G,f}$
        whose origins are adjacent.
        By Lemma~\ref{lem:edges-compressed-precompressed},
        $\set{\vertB,\vertC}$  is
        an edge in $\CFI{G,f}$ if and only if
        $\set{\vertB/_\compEquiv,\vertC/_\compEquiv}$ is an edge in
        $\compCFI{G,f}{\compEquiv}$ if and only if
        $\set{\autoA(\vertB/_\compEquiv),\autoA(\vertC/_\compEquiv)}$
        is an edge in $\compCFI{G,g}{\compEquiv}$ if and only if
        $\set{\autoB(\vertB),\autoB(\vertC)}$ is an edge in
        $\CFI{G,g}$.
        Two vertices, whose origins are not adjacent,
        are never adjacent themselves.  Because $\autoB$ is color-preserving, $\autoB$ is indeed an
        isomorphism from $\precompCFI{G,f}{\compEquiv}$ to
        $\precompCFI{G,g}{\compEquiv}$.  It satisfies
        $\autoB(\tup{\vertA}) = \tup{\vertA}$ because~$\auto$
        satisfies
        $\auto(\tup{\vertA}/_\compEquiv) = \tup{\vertA}/_\compEquiv$
        and no equivalence class contains two vertices of the same
        gadget.

	Assume there is an isomorphism $\auto \colon \precompCFI{G,f}{\compEquiv} \to \precompCFI{G,g}{\compEquiv}$ with $\auto(\tup{\vertA}) = \tup{\vertA}$.
	In particular,~$\autoA$ is an isomorphism $\CFI{G,f} \to \CFI{G,g}$
	corresponding to a $G$\nobreakdash-twisting~$T$ by Lemma~\ref{lem:cfi-iso}.
	We argue that the twisting~$T$ has to be $\compEquiv$\nobreakdash-compressible.
	Let $\vertB\in V(G)$ be an arbitrary base vertex and $\set{\vertB,\vertC} \in E(G)$ be an incident base edge.
	Also suppose that~$\vertC$ is the $i$\nobreakdash-th neighbor of~$\vertB$.
	If~$\autoA$ swaps the two sets $\setcond{(\vertB,\tup{a})}{\sum \tup{a} = 0, a_i = 0}$ and $\setcond{(\vertB,\tup{a})}{\sum \tup{a} = 0, a_i = 1}$,
	then~$\autoA$ has to swap the corresponding sets for all $\vertB' \compEquiv \vertB$.
	This means that, if $(\vertB,\vertC) \in T$, $\vertB' \compEquiv \vertB$,
	and~$\vertC'$ is the $i$\nobreakdash-th neighbor of~$\vertB'$,
	then also $(\vertB', \vertC') \in T$.
	Hence,~$T$ is $\compEquiv$\nobreakdash-compressible
	and satisfies $f = g + g_T$.
	Because we have $\auto(\tup{\vertA}) = \tup{\vertA}$,
	the isomorphism~$\auto$ is actually the identity on all vertices with origin~$W$ (every non-trivial automorphism of a gadget acts non-trivial on every vertex of that gadget).
	This means that~$T$ fixes~$W$.

	Assume there is a $\compEquiv$\nobreakdash-compressible $G$\nobreakdash-twisting~$T$ with $f = g + g_T$ and that fixes~$W$.
	We construct an isomorphism $\autoB\colon \compCFI{G,f}{\compEquiv} \to \compCFI{G,g}{\compEquiv}$.
	Let $\autoA \colon \CFI{G,f} \to \CFI{G,g}$ be the isomorphism corresponding to~$T$ from Lemma~\ref{lem:cfi-iso}.
	Similarly to the case before,
	one verifies that~$\autoA$ preserves~$\compEquiv$,
	that is,~$\autoA$ is an isomorphism from $\precompCFI{G,f}{\compEquiv}$ to $\precompCFI{G,g}{\compEquiv}$.
	Hence,~$\autoA$ maps a $\compEquiv$-equivalence class to a $\compEquiv$-equivalence class.
	This induces a map~$\autoB$ on the $\compEquiv$\nobreakdash-equivalence classes
	for which one easily shows that it is an isomorphism $\compCFI{G,f}{\compEquiv} \to \compCFI{G,g}{\compEquiv}$.
	Because~$T$ fixes~$W$, we have that $\autoA(\tup{\vertA}) = \tup{\vertA}$ and hence that~$\autoB(\tup{\vertA}/_\compEquiv) = \tup{\vertA}/_\compEquiv$.
\end{proof}

\subsection{The WL-Algorithm and Compressed CFI Graphs}

We compare the expressiveness of $k$\nobreakdash-WL (via the $(k+1)$-bijective pebble game)
on CFI graphs, precompressed CFI graphs, and compressed CFI graphs.
We again fix an arbitrary ordered base graph~$G$.

\begin{lemma}
	\label{lem:distinguish-CFI-graphs}
	For every $k \geq 3$, every $r \in \nat$,
	every $G$\nobreakdash-compression~$\compEquiv$,
	and all $\compEquiv$\nobreakdash-compressible functions $f,g\colon E(G) \to \FF_2$,
	\begin{enumerate}
		\item $\CFI{G,f} \not\kequivr{k}{r} \CFI{G,g}$ implies
		$\precompCFI{G,f}{\compEquiv} \not\kequivr{k}{r} \precompCFI{G,g}{\compEquiv}$,
		\item $\precompCFI{G,f}{\compEquiv} \not\kequivr{k}{r} \precompCFI{G,g}{\compEquiv}$ implies
		$\compCFI{G,f}{\compEquiv} \not\kequivr{k}{r} \compCFI{G,g}{\compEquiv}$, and
		\item $\compCFI{G,f}{\compEquiv} \not\kequivr{k}{r} \compCFI{G,g}{\compEquiv}$ implies
		$\precompCFI{G,f}{\compEquiv} \not\kequivr{k}{r+2} \precompCFI{G,g}{\compEquiv}$.
	\end{enumerate}
\end{lemma}

\begin{proof}
	Let $k \geq 3$, $r \in \nat$,~$\compEquiv$ be a $G$\nobreakdash-compression,
	and $f,g\colon E(G)\to \FF_2$ be $\compEquiv$\nobreakdash-compressible.
	\begin{enumerate}
		\item Because $\CFI{G,f}$ and $\CFI{G,g}$ are only extended by an additional relation, the claim directly follows.

		\item Assume that $\precompCFI{G,f}{\compEquiv} \not\kequivr{k}{r} \precompCFI{G,g}{\compEquiv}$,
		that is, Spoiler has a winning strategy in the $r$\nobreakdash-round bijective $k$\nobreakdash-pebble game played on $\precompCFI{G,f}{\compEquiv}$ and $\precompCFI{G,g}{\compEquiv}$.
		We show that Spoiler also has a winning strategy in the $r$\nobreakdash-round game
		on $\compCFI{G,f}{\compEquiv}$ and $\compCFI{G,g}{\compEquiv}$.

		We may assume without loss of generality that Duplicator always plays color-preserving in both games.
		This means that the bijections chosen by Duplicator always preserve colors.
		Observe that, if Duplicator does not do so, Spoiler wins immediately.

		Consider a position $\tup{\vertA}/_\compEquiv,\tup{\vertB}/_\compEquiv$ of the game on $\compCFI{G,f}{\compEquiv}$ and $\compCFI{G,g}{\compEquiv}$.
		That is, we have $\tup{\vertA}/_\compEquiv=(\vertA_1/_\compEquiv,\ldots,\vertA_\ell/_{\compEquiv})$ and $\tup{\vertB}/_\compEquiv=(\vertB_1/_\compEquiv,\ldots,\vertB_\ell/_{\compEquiv})$ for some $\ell\le k$ with pebbles~$p_{i_j}$ placed on $\vertA_{j}/_{\compEquiv}$ and the corresponding~$q_{i_j}$ placed on $\vertB_{j}/_{\compEquiv}$.
		We say that a position $\tup{\vertA}',\tup{\vertB}'$ of the game on $\precompCFI{G,f}{\compEquiv}$ and $\precompCFI{G,g}{\compEquiv}$ is an \defining{$s$-round witness} for
		$\tup{\vertA}/_\compEquiv,\tup{\vertB}/_\compEquiv$ if $\tup{\vertA}'=(\vertA_1',\ldots,\vertA_\ell')$ and
		$\tup{\vertB}'=(\vertB_1',\ldots,\vertB_\ell')$ such that the following conditions are satisfied:
		\begin{enumerate}[label=(\alph*), ref=(\alph*)]
			\item \label{cond:equiv} $\vertA_i'\compEquiv \vertA_i$ and $\vertB_i'\compEquiv \vertB_i$ for all $i\in[\ell]$;
			\item \label{cond:color} $\vertA'_i$ has the same color in $\precompCFI{G,f}{\compEquiv}$ as $\vertB'_i$ has in $\precompCFI{G,g}{\compEquiv}$ for every $i \in [\ell]$;
			\item \label{cond:win} $\tup{\vertA}',\tup{\vertB}'$ is a winning position for Spoiler in the $s$-round game on $\precompCFI{G,f}{\compEquiv}$ and $\precompCFI{G,g}{\compEquiv}$.
		\end{enumerate}
		We claim that if $\tup{\vertA}',\tup{\vertB}'$ is a $0$\nobreakdash-round witness for $\tup{\vertA}/_\compEquiv,\tup{\vertB}_\compEquiv$,
		then Spoiler wins the game on $\compCFI{G,f}{\compEquiv}$ and $\compCFI{G,g}{\compEquiv}$.
		Because $\tup{\vertA}',\tup{\vertB}'$ is a winning position for Spoiler in the $0$\nobreakdash-round game on $\precompCFI{G,f}{\compEquiv}$ and $\precompCFI{G,g}{\compEquiv}$
		by Condition~\ref{cond:win},
		the mapping $\tup{\vertA}'\mapsto\tup{\vertB}'$ is not a partial isomorphism.
		We show that the mapping
		$\tup{\vertA}/_\compEquiv
                \mapsto\tup{\vertB}/_\compEquiv$ is not a partial
                isomorphism, either.

		To see this,
		first note that $\vertA_i/_\compEquiv$ has the same color
		as $\vertB_i/_\compEquiv$
		and $\vertA'_i$ has the same color as $\vertB'_i$
		for every $i \in [\ell]$
		because Duplicator plays color-preserving and because of Condition~\ref{cond:color}.
		Suppose $i,j\in[\ell]$.
		If $\vertA_i'\compEquiv \vertA_j'$ but $\vertB_i'\not\compEquiv  \vertB_j'$, then
		by Condition~\ref{cond:equiv}, $\vertA_i/_{\compEquiv}=\vertA_j/_\compEquiv$ but $\vertB_i/_\compEquiv\neq \vertB_j/_\compEquiv$.
		So Spoiler wins immediately.
		If $\vertA_i'=\vertA_j'$ but $\vertB_i'\neq \vertB_j'$,
		then $\vertB_i'$ has the same color as $\vertB_j'$ (namely the one of $\vertA_i'$)
		and $\vertB_i' \not\compEquiv\vertB_j'$
		because distinct vertices of the same color are never $\compEquiv$-\nobreakdash equivalent.
		Thus, $\vertA_i/_{\compEquiv}=\vertA_j/_\compEquiv$ but $\vertB_i/_\compEquiv\neq \vertB_j/_\compEquiv$ by Condition~\ref{cond:equiv}.
		Spoiler wins immediately again.

		Lastly, suppose that $\set{\vertA_i',\vertA_j'}$ is an edge in $\precompCFI{G,f}{\compEquiv}$
		but $\set{\vertB_i',\vertB_j'}$
		is not an edge in $\precompCFI{G,g}{\compEquiv}$. Then $\set{\vertA_i/_\compEquiv,\vertA_j/_\compEquiv}$ is an edge in $\compCFI{G,f}{\compEquiv}$.
		We prove that $\set{\vertB_i/_\compEquiv,\vertB_j/_\compEquiv}$ is not an edge in $\compCFI{G,g}{\compEquiv}$ and hence Spoiler wins.
		Because $\set{\vertA_i', \vertA_j'}$ is an edge,
		the origins of~$\vertA_i'$ and~$\vertA_j'$ are adjacent base vertices.
		Hence, the origins of~$\vertB_i'$ and~$\vertB_j'$ are also adjacent base vertices,
		because they have the same colors as~$\vertA_i'$ and~$\vertA_j'$, respectively.
		So by Lemma~\ref{lem:edges-compressed-precompressed},
		the set $\set{\vertB_i',\vertB_j'}$
		is not an edge in $\precompCFI{G,g}{\compEquiv}$.

		The same arguments apply in the converse direction if the relations hold between the $\vertB_i'$, but not the $\vertA_i'$.

		By induction on~$s\leq r$ we shall prove that Spoiler has a strategy for the $r$\nobreakdash-round game on $\compCFI{G,f}{\compEquiv}$ and $\compCFI{G,g}{\compEquiv}$ such that the position in round~$s$ has an $(r-s)$\nobreakdash-round witness.
		The base case $s=0$ holds because the empty position in the game on $\precompCFI{G,f}{\compEquiv}$ and $\precompCFI{G,g}{\compEquiv}$ is an $r$\nobreakdash-round witness for the empty position in the game on  $\compCFI{G,f}{\compEquiv}$ and $\compCFI{G,g}{\compEquiv}$.
		For the inductive step, let $s<r$ and suppose that the position $\tup{\vertA}/_\compEquiv,\tup{\vertB}_\compEquiv$ in round~$s$ has an $(r-s)$\nobreakdash-witness $\tup{\vertA}',\tup{\vertB}'$.
		If Spoiler removes a pair of pebbles, then we can remove the corresponding elements for the tuples, and the witness relation remains intact.
		Now suppose Spoiler wants to place a pair of pebbles $p_{i_{\ell+1}},q_{i_{\ell+1}}$.
		Duplicator picks a bijection~$h$ from $V(\compCFI{G,f}{\compEquiv})$ to $V(\compCFI{G,g}{\compEquiv})$.
		Without loss of generality,
		we may assume that $h(\vertA_i/_{\compEquiv})=\vertB_i/_{\compEquiv}$ for all~$i \in [\ell]$.
		Also recall that Duplicator plays color-preserving.
		This implies that~$h$ permutes each color class.
		Hence, we can lift~$h$ to a bijection~$h'$ from $V(\precompCFI{G,f}{\compEquiv})$ to $V(\precompCFI{G,f}{\compEquiv})$:
		We set $h'(\vertA') = \vertB'$ if and only if $\vertA'$ and $\vertB'$ are of the same origin and $h(\vertA'/_\compEquiv) = \vertB'_\compEquiv$.
		This, in particular, implies  $h'(\vertA_i')=\vertB_i'$ for all~$i \in [\ell]$.

		Now suppose that in the $(r-s)$\nobreakdash-round game on $\precompCFI{G,f}{\compEquiv}$ and $\precompCFI{G,g}{\compEquiv}$ in position~$\tup{\vertA}',\tup{\vertB}'$, Duplicator picks the bijection~$h'$, and Spoiler answers (according to Spoiler's winning strategy) by placing the pebble $p_{i_{\ell+1}}$ on the vertex $\vertA_{\ell+1}'$ and the pebble $q_{i_{\ell+1}}$ on the vertex ${\vertB_{\ell+1}'\coloneqq h'(\vertA_{\ell+1}')}$. Then in the game on $\compCFI{G,f}{\compEquiv}$ and $\compCFI{G,g}{\compEquiv}$ with Duplicator choosing the bijection~$h$,
		Spoiler places $p_{i_{\ell+1}}$ on $\vertA_{\ell+1}'/_{\compEquiv}$ and
		$q_{i_{\ell+1}}$ on $\vertB_{\ell+1}'/_{\compEquiv}$.
		Then the pair
		$(\vertA_1',\ldots,\vertA_{\ell+1}'),(\vertB_1',\ldots,\vertB_{\ell+1}')$ is an $(r-s-1)$\nobreakdash-round witness for the new position
		$(\vertA_1/_{\compEquiv},\ldots,\vertA_\ell/_{\compEquiv},\vertA_{\ell+1}'/_{\compEquiv})$, $(\vertB_1/_{\compEquiv},\ldots,\vertB_\ell/_{\compEquiv},\vertB_{\ell+1}'/_{\compEquiv})$.

		\item Assume that $\compCFI{G,f}{\compEquiv} \not\kequivr{k}{r} \compCFI{G,g}{\compEquiv}$.
		We turn a winning strategy of Spoiler in the $r$\nobreakdash-round bijective $k$\nobreakdash-pebble game played on $\compCFI{G,f}{\compEquiv}$ and $\compCFI{G,g}{\compEquiv}$
		into a winning strategy of Spoiler in the $(r+2)$\nobreakdash-round bijective
		$k$\nobreakdash-pebble game played on $\precompCFI{G,f}{\compEquiv}$ and $\precompCFI{G,g}{\compEquiv}$.

		We call a position $\tup{\vertA},\tup{\vertB}$
		of the game played on $\precompCFI{G,f}{\compEquiv}$ and $\precompCFI{G,g}{\compEquiv}$,
		that is, ${\tup{\vertA}=(\vertA_1, \dots, \vertA_\ell)}$
		and ${\tup{\vertB} = (\vertB_1, \dots, \vertB_\ell)}$ for some $\ell \leq k$,
		\defining{$s$\nobreakdash-round witnessed}
		if the position \[\tup{\vertA}/_\compEquiv=(\vertA_1/_\compEquiv, \dots,\vertA_\ell/_\compEquiv) ,\tup{\vertB}/_\compEquiv=(\vertB_1/_\compEquiv, \dots,\vertB_\ell/_\compEquiv)\]
		of the game on $\compCFI{G,f}{\compEquiv}$ and $\compCFI{G,g}{\compEquiv}$ is a winning position for Spoiler in the $s$\nobreakdash-round game on $\compCFI{G,f}{\compEquiv}$ and $\compCFI{G,g}{\compEquiv}$.

		We show that if $\tup{\vertA},\tup{\vertB}$ is $0$-round witnessed,
		then Spoiler wins the game on $\precompCFI{G,f}{\compEquiv}$ and $\precompCFI{G,g}{\compEquiv}$ in at most $2$ additional rounds.
		The position $\tup{\vertA}/_\compEquiv, \tup{\vertB}/_\compEquiv$ in the game on $\compCFI{G,f}{\compEquiv}$ and $\compCFI{G,g}{\compEquiv}$ does not induce a partial isomorphism.
		Because~$\vertA_i$ and~$\vertB_i$ have the same color (Duplicator plays color-preserving),~$\vertA_i/_\compEquiv$ and~$\vertB_i/_\compEquiv$ have the same color for all $i \in[\ell]$.
		If $\vertA_i/_\compEquiv=\vertA_j/_\compEquiv$ but $\vertB_i/_\compEquiv\neq \vertB_j/_\compEquiv$,
		then  $\vertA_i\compEquiv \vertA_j$ but
		$\vertB_i \not\compEquiv \vertB_j$ and Spoiler wins immediately.
		Suppose $\set{\vertA_i/_\compEquiv,\vertA_j/_\compEquiv}$ is an edge in $\compCFI{G,f}{\compEquiv}$
		but $\set{\vertB_i/_\compEquiv,\vertB_j/_\compEquiv}$ is not an edge in $\compCFI{G,g}{\compEquiv}$.
		On the one hand, there are vertices $\vertA'_i\compEquiv\vertA_i$ and $\vertA'_j\compEquiv\vertA_j$
		such that $\set{\vertA'_i,\vertA'_j}$ is an edge in $\precompCFI{G,f}{\compEquiv}$.
		On the other hand, for every $\vertB'_i\compEquiv\vertB_i$ and $\vertB'_j\compEquiv\vertB_j$,
		the set $\set{\vertB'_i,\vertB'_j}$ is not an edge in $\precompCFI{G,f}{\compEquiv}$.
		Spoiler picks up a pebble pair different from the two placed on $\vertA_i,\vertB_i$  and $\vertA_j,\vertB_j$
		(such a pair exists because $k \geq 3$).
		Spoiler places one pebble on~$\vertA_i'$ and the other one on some vertex~$\vertB_i'$ according to Duplicator's bijection.
		If $\vertB_i' \not \compEquiv \vertB_i$, then Spoiler wins.
		Otherwise, Spoiler picks up the pebble pair placed on~$\vertA_j$ and~$\vertB_j$.
		Spoiler places one pebble on~$\vertA_j'$ and the other one on some vertex~$\vertB_j'$ according to Duplicator's bijection.
		If $\vertB_j' \not \compEquiv \vertB_j$,
		then Spoiler wins again.
		Otherwise, as already argued above, $\set{\vertB_i',\vertB_j'}$ is not an edge, but $\set{\vertA_i',\vertA_j'}$ is.
		Thus, Spoiler wins after $2$ additional rounds.

		The same arguments apply in the converse direction if the relations hold between the $\vertB_i$, but not the $\vertA_i$.

		We prove by induction on~$s\leq r$ that Spoiler has a winning strategy
		in the ${(r+2)}$\nobreakdash-round game on $\precompCFI{G,f}{\compEquiv}$ and $\precompCFI{G,g}{\compEquiv}$ such that the position reached in round~$s$ is ${(r-s)}$\nobreakdash-round witnessed.
		Clearly, the initial position is $r$\nobreakdash-round witnessed because
		Spoiler wins the $r$\nobreakdash-round game $\compCFI{G,f}{\compEquiv}$ and $\compCFI{G,g}{\compEquiv}$.
		Assume $s < r$ and assume, by the induction hypothesis,
		that the position $\tup{\vertA},\tup{\vertB}$
		of the game on $\precompCFI{G,f}{\compEquiv}$ and $\precompCFI{G,g}{\compEquiv}$
		is $(r-s)$\nobreakdash-round witnessed.
		If Spoiler removes a pair of pebbles, then we can remove the corresponding elements from the tuples, and the position is still witnessed.
		Now suppose Spoiler wants to place a pair of pebbles $p_{i_{\ell+1}},q_{i_{\ell+1}}$.
		Duplicator picks a color-preserving bijection~$h$ from $V(\precompCFI{G,f}{\compEquiv})$ to $V(\precompCFI{G,g}{\compEquiv})$.
		We construct a bijection~$h'$ from $V(\compCFI{G,f}{\compEquiv})$ to $V(\compCFI{G,g}{\compEquiv})$.
		Assume~$\vertC/_\compEquiv$ is an equivalence class
		such that~$\vertC$ is the unique minimal vertex of that class (with respect to the colors).
		We define $h'(\vertC/_\compEquiv) \coloneqq h(\vertC)/_\compEquiv$.
		Because minimal vertices in equivalence classes are unique,
		because vertices of the same gadget are never in the same equivalence class,
		and because~$h$ is color preserving,
		the map~$h'$ is indeed a bijection.
		Now suppose that in the $(r-s)$\nobreakdash-round game on $\compCFI{G,f}{\compEquiv}$ and $\compCFI{G,g}{\compEquiv}$,
		Duplicator plays~$h'$ as bijection.
		According to the winning strategy, Spoiler places the pebble~$p_{i_{\ell+1}}$ on~$\vertA_{\ell+1}/_\compEquiv$ and~$q_{i_{\ell+1}}$ on $\vertB_{\ell+1}/_\compEquiv = h'(\vertA_{\ell+1}/_\compEquiv)$,
		where we assume without loss of generality that~$\vertA_{\ell+1}$ and $\vertB_{\ell+1}$ are the minimal vertices of their classes.
		In the game on $\precompCFI{G,f}{\compEquiv}$ and $\precompCFI{G,g}{\compEquiv}$,
		where Duplicator chooses the bijection $h$,
		Spoiler places~$p_{i_{\ell+1}}$ on~$\vertA_{\ell+1}$
		and~$q_{i_{\ell+1}}$ on $\vertB_{\ell+1} = h(\vertA_{\ell+1})$.
		Because Spoiler has a winning strategy on $\compCFI{G,f}{\compEquiv}$ and $\compCFI{G,g}{\compEquiv}$,
		the position
		$(\vertA_1,\ldots,\vertA_{\ell+1}),(\vertB_1,\ldots,\vertB_{\ell+1})$
		is $(r-s-1)$\nobreakdash-round witnessed.\qedhere
	\end{enumerate}
\end{proof}

\begin{corollary}
	\label{cor:bounds-precompressed-compressed}
	For every $k \geq 3$, every $r \in \nat$, every $G$\nobreakdash-compression $\compEquiv$,
	and all $\compEquiv$\nobreakdash-compressible $f,g\colon E(G)  \to \FF_2$, we have
	\begin{align*}
		&\text{if } &\precompCFI{G,f}{\compEquiv} &\kequivr{k}{r} \precompCFI{G,g}{\compEquiv} &\text{ and }& & \precompCFI{G,f}{\compEquiv} &\not\kequivr{k}{r+1} \precompCFI{G,g}{\compEquiv},\\
		&\text{then } & \compCFI{G,f}{\compEquiv} &\kequivr{k}{r-2} \compCFI{G,g}{\compEquiv} &\text{ and }& & \compCFI{G,f}{\compEquiv} &\not\kequivr{k}{r+1} \compCFI{G,g}{\compEquiv}.
	\end{align*}
\end{corollary}

By the former corollary, we can study precompressed CFI graphs to obtain lower bounds on the iteration number of $k$\nobreakdash-WL on compressed CFI graphs.
To do so, we introduce a variant of the Cops and Robber game suitable for compressions.

\subsection{Cops and Robbers for Precompressed CFI Graphs}

We describe a variant of the $k$-Cops and Robber game suitable for compressed CFI graphs.
The \defining{compressed $k$-Cops and Robber game} is played on an ordered base graph~$G$ and a $G$\nobreakdash-compression~$\compEquiv$.
The cops player places cops on up to~$k$ many $\compEquiv$\nobreakdash-equivalence classes and the robber is placed on one edge of~$G$.
Initially, only the robber is placed.
The game proceeds as follows:
\begin{enumerate}
	\item \label{itm:cop-picked-up-compressed} One cop is picked up, and a destination $\compEquiv$-equivalence class~$c$ for this cop is selected.
	\item The robber moves. To move from the current edge~$e_1$ to another edge~$e_2$, the robber has to provide a $\compEquiv$\nobreakdash-compressible $G$\nobreakdash-twisting
	that only twists the edges~$e_1$ and~$e_2$ and that fixes every vertex contained in a $\compEquiv$\nobreakdash-equivalence class occupied by a cop.
	\item  The cop picked up in Step~\ref{itm:cop-picked-up-compressed} is placed on $c$.
\end{enumerate}
The robber is caught if the two endpoints of the robber-occupied edge are contained in cop-occupied $\compEquiv$\nobreakdash-equivalence classes.
If the robber is caught after~$r$~rounds,
then the cops win in~$r$~rounds.
Otherwise, the robber wins in~$r$~rounds.
Note that the initial position of the robber may matter to decide who wins in the compressed game.

\begin{lemma}
	\label{lem:compressed-cops-and-robber}
	Let $k \geq 3$,
	$\compEquiv$ be a $G$\nobreakdash-compression,
	and $f, g \colon E(G) \to \FF_2$ such that there is exactly one twisted edge~$e$ with respect to $f$ and $g$.
	If the robber has a winning strategy in the $r$\nobreakdash-round
	compressed ${k}$\nobreakdash-Cops and Robber game,
	where the robber is initially placed on~$e$,
	then $\precompCFI{G,f}{\compEquiv} \kequivr{k}{r} \precompCFI{G,g}{\compEquiv}$.
\end{lemma}

\begin{proof}
	We show that
	Duplicator has a winning strategy in the $r$\nobreakdash-round bijective $k$\nobreakdash-pebble game.
	Duplicator maintains a function $g' \colon E(G) \to \FF_2$
	and an edge $e' \in E(G)$
	such that
	after $s\leq r$ rounds in position $\tup{\vertA},\tup{\vertB}$,
	\begin{enumerate}[label=(\alph*), ref=(\alph*)]
		\item \label{cond:some-pebbles}
		there is an isomorphism $\autoA \colon \precompCFI{G,g}{\compEquiv} \to \precompCFI{G,g'}{\compEquiv}$
		such that $\autoA(\tup{\vertB}) =\tup{\vertA}$ (recall that CFI graphs over the same base graph have the same vertex set),
		\item \label{cond:edge} only the base edge~$e'$ is twisted with respect to~$f$ and~$g'$,
		\item \label{cond:twist-not-caught}  at most one endpoint of~$e'$ is the origin of a vertex in $\tup{\vertA}$, and
		\item \label{cond:robber-wins} the robber has a winning strategy in the $(r - s)$\nobreakdash-round compressed $k$\nobreakdash-Cop and Robber game, where the robber is placed on~$e'$ and the cops on the equivalence classes of the origins of the vertices in~$\tup{\vertA}$.
	\end{enumerate}
	Because initially no pebbles are placed, the invariant clearly holds
	for $g' = g$ and $e' = e$.
	Assume, by the induction hypothesis, that after $s < r$ many rounds the invariant holds.
	Spoiler picks up a pair of pebbles.
	By Condition~\ref{cond:some-pebbles}, these pebbles were placed on vertices with the same origin.
	In the compressed Cop and Robber game, a cop is picked up from the equivalence class of this origin.
	Now assume that Spoiler wants to place a pair of pebbles $p_{i_\ell+1},q_{i_\ell+1}$.
	Duplicator defines a bijection~$h$ from $V(\precompCFI{G,f}{\compEquiv})$ to $V(\precompCFI{G,g}{\compEquiv})$ as follows:
	Let~$\vertC$ be a vertex of $\precompCFI{G,f}{\compEquiv}$ and we want to define the image~$h(\vertC)$.
	Consider the compressed Cop and Robber game where a cop is placed on the equivalence class of the origin of~$\vertC$.
	Let $T_\vertC$ be the $\compEquiv$\nobreakdash-compressible $G$\nobreakdash-twisting by which the robber moves from~$e'$ to~$e'_\vertC$ following the robber's winning strategy.
	Furthermore, let $\autoB_\vertC\colon\precompCFI{G,g'}{\compEquiv} \to \precompCFI{G,g'_\vertC}{\compEquiv}$ be the isomorphism corresponding to $T_\vertC$
	by Lemma~\ref{lem:twistsings-equiv-parameters}.
	The isomorphism~$\autoB_\vertC$ is the identity on all vertices whose origin is fixed by~$T_\vertC$.
	In particular,~$\autoB_\vertC$ is the identity on the pebbled vertices
	because their origins are occupied by the cops.
	Duplicator defines $h(\vertC) \coloneqq \inv{\autoA}(\inv{\autoB_\vertC}(\vertC))$.
	The map~$h$ is clearly color-preserving.
	To prove that~$h$ is a bijection,
	it suffices to show that~$h$ permutes every color class.
	Because the move of the robber only depends on the origin of~$\vertC$,
	we actually have $T_\vertC = T_{\vertC'}$ for all $\vertC,\vertC'$
	with the same origin or, equivalently, of the same color.
	Hence, $\autoB_\vertC = \autoB_\vertC'$ for all $\vertC,\vertC'$ of the same color and $h$ is a bijection.

	Now Spoiler picks a vertex~$\vertC$ of $\precompCFI{G,f}{\compEquiv}$
	and places the pebble $p_{i_\ell+1}$ on~$\vertC$ and the pebble $q_{i_\ell+1}$ on $h(\vertC)=\inv{\autoA}(\inv{\autoB_\vertC}(\vertC))$.
	Set $g'' \coloneqq g'_\vertC$ and $e'' \coloneqq  e'_\vertC$.
	Because~$e'$ is the only edge twisted with respect to $f$ and~$g'$ and~$T_\vertC$ twists~$e'$ and~$e_\vertC$,
	the edge~$e'' = e_\vertC'$ is the only edge twisted with respect to~$f$ and~$g''$.
	So Condition~\ref{cond:edge} holds.
	Because~$\autoA$ is an isomorphism from
	$\precompCFI{G,g}{\compEquiv}$ to $\precompCFI{G,g'}{\compEquiv}$
	and $\autoB_\vertC$ is an isomorphism from
	$\precompCFI{G,g'}{\compEquiv}$ to $\precompCFI{G,g''}{\compEquiv}$,
	the map $\autoA' \coloneqq \autoB_\vertC \circ \autoA$
	is an isomorphism from $\precompCFI{G,g}{\compEquiv}$ to $\precompCFI{G,g'}{\compEquiv}$.
	Hence, $\autoA'(h(\vertC)) = \autoB_\vertC(\autoA(\inv{\autoA}(\inv{\autoB_\vertC}(\vertC)))) = \vertC$.
	Because~$T_\vertC$ fixes the origins of all vertices in~$\tup{\vertA}$,
	the isomorphism~$\autoB_\vertC$ is the identity on~$\tup{\vertA}$.
	Thus, Condition~\ref{cond:some-pebbles} holds.
	Condition~\ref{cond:twist-not-caught} holds because the twisting~$T_\vertC$ was given by a strategy of the robber.
	By placing the cop on the equivalence class of~$\vertC$ and moving the robber accordingly, the robber has a winning strategy in $r-s-1$ rounds
	by Condition~\ref{cond:robber-wins}.
	So the invariant is satisfied after round $s+1$.
	To show that Duplicator has not lost in this round,
	we need to argue that the new position defines a partial isomorphism.
	Corresponding pebbles are placed on vertices of the same color,
	and if $p_i$ and $p_j$ are placed on the same vertex,
	then so are $q_i$ and $q_j$ by Condition~\ref{cond:some-pebbles}.
	By Condition~\ref{cond:twist-not-caught},
	edges between two adjacent origins of pebbled vertices
	are never twisted with respect to $f$ and $g''$.
	So between $\precompCFI{G,f}{\compEquiv}$ and $\precompCFI{G,g''}{\compEquiv}$,
	the identity map on the pebbled vertices in $\precompCFI{G,f}{\compEquiv}$
	is a partial isomorphism between $\precompCFI{G,f}{\compEquiv}$ and $\precompCFI{G,g''}{\compEquiv}$.
	By Condition~\ref{cond:some-pebbles},
	the pebbled vertices also define a partial isomorphism between $\precompCFI{G,f}{\compEquiv}$ and $\precompCFI{G,g}{\compEquiv}$.
	Duplicator updates $g' \leftarrow g''$ and $e' \leftarrow e''$.
\end{proof}

From now on, we focus on compressible twistings and the compressed Cops and Robber game and do not require the details of the CFI construction anymore.

\section{A Lower Bound for the Iteration Number}
\label{sec:lower-bound}

In this section, we give the proof of Theorem \ref{thm:main}.
We start by describing the base graphs and the equivalence relations that are used to construct the compressed CFI graphs.
We also describe several auxiliary objects that are relevant for the analysis.

For the remainder of this section, let us fix an arbitrary integer~$k \geq 3$.
We define
\begin{equation}
	f(k) \coloneqq 2k+2.
\end{equation}
Also let $w$ be an integer such that
\begin{equation}
	\label{eq:w-bound}
	w \geq \max\big(2 \cdot f(k),4k \cdot\ln(4k)\big),
\end{equation}
and let~$p_0,\ldots,p_{k-1}$ be pairwise coprime integers such that
\begin{equation}
	\frac{w}{2} < p_i \leq w
\end{equation}
for all~$i \in [0,k-1]$.
Since $w$ is sufficiently large, such a collection of coprime integers exists.
In fact, we can even choose them to be distinct prime numbers~\cite{Ramanujan19}.
More precisely, we can rely on the following result.

\begin{lemma}[\cite{Sondow09}]
	\label{lem:primes}
	For all integers $k' \geq 1$ and all integers $w \geq 4k \ln(4k)$, there are distinct primes $p_0,\dots,p_{k'-1}$ such that $\frac{w}{2} < p_i \leq w$ for all $i \in [0,k'-1]$.
\end{lemma}

\begin{remark}
	At this point, one may wonder whether a better dependency between $w$ and $k$ can be achieved if, instead of prime numbers, we choose pairwise coprime numbers $p_0,\dots,p_{k-1}$.
	Suppose that $k \geq 2$ and $w \geq k$ such that there are pairwise coprime integers $p_0,\ldots,p_{k-1}$ with $\frac{w}{2} < p_i \leq w$ for all~$i \in [0,k-1]$.
	Since $p_i \geq 2$ for all~$i \in [0,k-1]$, we can pick an arbitrary prime factor $q_i$ of $p_i$ for every~$i \in [0,k-1]$.
	Hence, we obtain distinct prime numbers $q_0,\ldots,q_{k-1}$ with $q_i \leq w$ for all~$i \in [0,k-1]$.
	So $\pi(w) \geq k$, where $\pi$ is the prime-counting function (i.e., $\pi(w)$ denotes the number of primes $q \leq w$).
	By the Prime Number Theorem, $\pi(w) \sim \frac{w}{\ln w}$.
	While this leaves room for some small improvements (over the bound from Lemma \ref{lem:primes}), it is not possible to obtain a linear dependence between $k$ and $w$.
\end{remark}

We use cylindrical grids as base graphs.
We fix \[I \coloneqq [0,k-1]\] as the set of rows.
Let~$\ell$ be a positive integer and suppose $J = [0,\ell-1]$.
We use~$I$ as set of rows and~$J$ as set of columns to construct the
\defining{cylindrical grid}~$C_{I,J}$ with vertex set $V(C_{I,J}) \coloneqq I \times J$.
Starting from the $I \times J$ grid, we connect the top and bottom most vertex in every column, i.e.,
we add the edge $\set{(0,j),(k-1,j)}$ for all~$j\in J$.
We also consider the \defining{toroidal grid}~$C_{I,J}^*$ with the same vertex set $V(C_{I,J}^*) \coloneqq I \times J$.
Starting from the cylindrical grid~$C_{I,J}$, we also connect the first and last vertex in every row, i.e.,
we add the edge $\set{(i,0), (i,\ell-1)}$ for all $i \in I$.
We turn~$C_{I,J}$ and~$C_{I,J}^*$ into ordered base graphs using the
lexicographical order on $I \times J$.

We fix
\begin{align*}
	J &\coloneqq [0,\frac{1}{2} \cdot f(k) \cdot p_1 \cdot p_2 \cdots p_k-1] \text{ and}\\
	J^* &\coloneqq [0,f(k) \cdot p_1 \cdot p_2 \cdots p_k-1].
\end{align*}
We consider the cylindrical grid $C \coloneqq C_{I,J}$ and the toroidal grid $C^* \coloneqq C_{I,J^*}^*$.
Note that~$C$ is a subgraph of~$C^*$ (see also~Figure~\ref{fig:grids}).
We use the cylindrical grid~$C$ as the actual base graph whereas~$C^*$ serves as an auxiliary graph in the analysis.

\newcommand{\drawGridConn}[3]{
	\pgfmathsetmacro{\w}{#1-1}
	\pgfmathsetmacro{\h}{#2-1}
	\foreach \n in {0,...,\h} {
		\draw[#3] (0,\n) to (\w,\n);
	}
}
\newcommand{\drawCycGrid}[3]{
	\pgfmathsetmacro{\w}{#1-1}
	\pgfmathsetmacro{\h}{#2-1}
	\foreach \n in {0,...,\w} {
		\draw[#3] (\n, 0) to (\n, \h);
		\draw[#3] (\n, \h) .. controls (\n+0.5,\h+2) and (\n+0.5,-2) .. (\n,0);
		\foreach \m in {0,...,\h} {
			\draw[#3, fill] (\n,\m) circle (0.15cm);
		}
	}
	\drawGridConn{#1}{#2}{#3}
}

\begin{figure}
	\centering
	\begin{tikzpicture}[scale=0.36]
		\def\height{3}
		\def\width{7}
		\def\sep{2}

		\pgfmathsetmacro{\th}{\height-1}
		\pgfmathsetmacro{\tm}{2*\width+1*\sep-1}
		\pgfmathsetmacro{\tw}{4*\width+2*\sep-1}
		\draw[use as bounding box, draw=none] (-4, -2.2) rectangle (\tw + 4, \th +4);

		\draw[blue!30!white, fill, rounded corners]
		(-0.3,-0.5) rectangle (\height+1+0.5,\height-1+0.5);
		\draw[blue!30!white, fill, rounded corners]
		(2*\width+\sep-\height-2-0.3,-0.5) rectangle (2*\width+\sep-1+0.5,\height-1+0.5);

		\draw [decorate, decoration = {brace}](2*\width+\sep-\height-2-0.3,\height-1+0.8) -- (2*\width+\sep-1+0.3,\height-1+0.8) node[pos=0.5,black, above=0] {\strut\footnotesize $k+2$};

		\begin{scope}[shift={(2*\width+\sep-1,0)}]
			\drawCycGrid{\width+1}{\height}{red}
		\end{scope}
		\begin{scope}[shift={(3*\width+\sep-1,0)}]
			\drawGridConn{\sep+2}{\height}{red, dashed}
		\end{scope}
		\begin{scope}[shift={(3*\width+2*\sep,0)}]
			\drawCycGrid{\width}{\height}{red}
		\end{scope}

		\foreach \n in {0,...,\th} {
			\draw[red] (0,\n) .. controls (-18,\n+\height+1.5) and (\tw+18,\n+\height+1.5) .. (\tw,\n);
		}

		\drawCycGrid{\width}{\height}{black}
		\begin{scope}[shift={(\width-1,0)}]
			\drawGridConn{\sep+2}{\height}{black, dashed}
		\end{scope}
		\begin{scope}[shift={(\width+\sep,0)}]
			\drawCycGrid{\width}{\height}{black}
		\end{scope}

		\draw [decorate, decoration = {brace}] (\tm+0.3,-0.7) -- (-0.3,-0.7) node[pos=0.5,black, below=0] {\strut\footnotesize $|J|$};
		\draw [decorate, red, decoration = {brace}] (\tw+0.3,-0.7) -- (\tm+1-0.3,-0.7) node[pos=0.5,black, below=0, red] {\strut\footnotesize $|J|$};
	\end{tikzpicture}
	\caption{The cylindrical grid $C=C_{I,J}$ is drawn in black.
		It is a subgraph of the toroidal grid $C^*=C^*_{I,J^*}$ for which the additional vertices and edges are drawn in red.
		The blue regions mark the areas of $C$ in which the robber can be located in the proof of Lemma~\ref{ref:lem-robber-lower-bound:new:construction}.}
	\label{fig:grids}
\end{figure}
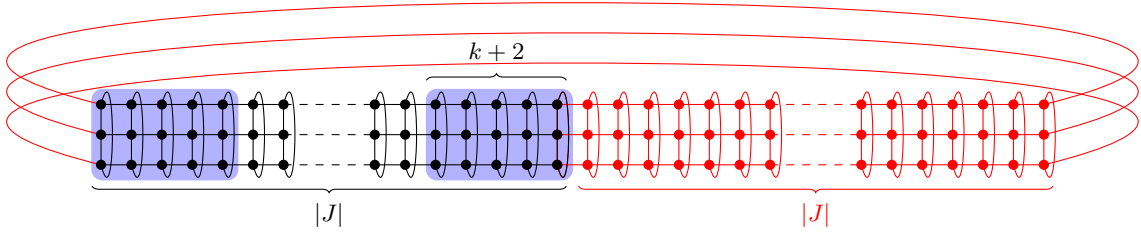

We define multiple equivalence relations on $V(C)$ and $V(C^*)$.
For a graph~$G$, an equivalence relation~$R$ on $V(G)$, and a vertex $\vertA \in V(G)$,
we denote by $\eqclass{R}{\vertA}$ the $R$-equivalence class containing~$\vertA$.
For a set $W \subseteq V(G)$, we define the $R$-class of $W$ via $\eqclass{R}{W} \coloneqq \bigcup_{\vertA \in W} \eqclass{R}{\vertA}$.

\colorlet{equivclassB}{colD}
\colorlet{equivclassA}{colA}
\colorlet{equivclassE}{colC}
\colorlet{equivclassD}{colF}
\colorlet{equivclassC}{colE}
\colorlet{equivclassF}{colB}
\colorlet{equivclassG}{colH}

\newcommand{\ponesize}{15}
\newcommand{\ptwosize}{35}
\newcommand{\pthreesize}{21}
\def\blocklength{{0,15,35,21}}
\def\offsetsmiddle{{0,8,5,12}}
\newcommand{\horizontalscale}{0.2}
\pgfmathparse{\blocklength[0]}\def\length{\pgfmathresult}
\newcommand{\theendoffirsthalf}{42*\horizontalscale}
\newcommand{\thebeginingofsecondhalf}{48*\horizontalscale}
\newcommand{\theendofthisgrid}{70*\horizontalscale}

\newcommand{\arcadedotsnew}[2]{
	\foreach \i in {-0.15,0,0.15} {
		\draw [fill=black] ($(#1,0)!0.5!(#2,0)+(\i,-0.75)$) circle (0.15mm);
	}

}

\tikzstyle{gridnode} = [circle, fill, inner sep=0, minimum size=1.5mm, minimum width = 1.5mm, minimum height =1.5mm]

\renewcommand{\ponesize}{11}
\renewcommand{\ptwosize}{18}
\renewcommand{\pthreesize}{16}
\renewcommand{\blocklength}{{0,11,18,16}}
\def\offsetsmiddle{{0,8,6,12}}
\renewcommand{\theendoffirsthalf}{26*\horizontalscale}
\renewcommand{\thebeginingofsecondhalf}{30.5*\horizontalscale}
\renewcommand{\theendofthisgrid}{51*\horizontalscale}
\newcommand{\theendofthirdhalf}{60*\horizontalscale}

\newcommand{\thebeginoffourthhalf}{64.5*\horizontalscale}
\newcommand{\theendoffourthhalf}{75.5*\horizontalscale}

\begin{figure}
	\centering
	\begin{tikzpicture}[font=\small, scale=0.85]

		\draw[use as bounding box, draw=none] (-0.6,0.8) rectangle (\theendoffourthhalf+0.1, -2.3);

		\foreach \i in {1,...,3} {
			\pgfmathsetmacro\thelengthhere{\blocklength[\i]}
			\pgfmathsetmacro\offsethere{\offsetsmiddle[\i]}
			\fill[black!10!white] (0,-\i*0.5+0.5) rectangle (\thelengthhere*0.5*\horizontalscale,-\i*0.5);
			\fill[black!10!white] (\theendofthisgrid-\offsethere*0.5*\horizontalscale-\thelengthhere*0.5*\horizontalscale,-\i*0.5+0.5) rectangle (\theendofthisgrid,-\i*0.5);
		}

		\draw (0,0) -- (0,-3*0.5);
		\draw (\theendofthisgrid,0) -- (\theendofthisgrid,-3*0.5);
		\draw[red] (\theendoffourthhalf,0) -- (\theendoffourthhalf,-3*0.5);

		\foreach \i in {0,...,3} {
			\draw (\theendoffirsthalf,-\i*0.5) -- (0,-\i*0.5) ;
			\draw (\thebeginingofsecondhalf,-\i*0.5) -- (\theendofthisgrid,-\i*0.5) ;
			\draw[red] (\theendofthisgrid,-\i*0.5) --(\theendofthirdhalf,-\i*0.5) ;
			\draw[red] (\thebeginoffourthhalf,-\i*0.5) --(\theendoffourthhalf,-\i*0.5) ;

		}

		\foreach \a/\z in {1/4,2/2,3/3}
		{
			\pgfmathparse{\blocklength[\a]}\def\thelengthhere{\pgfmathresult}
			\foreach \x in {1,...,\z}
			{

				\draw ($(\thelengthhere*\horizontalscale*0.5*\x,0.5-0.5*\a)$) -- ($(\thelengthhere*\horizontalscale*0.5*\x,0.5-0.5*\a-0.5)$);
			}
		}
		\foreach \a/\z in {1/2,2/1,3/1}
		{
			\pgfmathsetmacro\thelengthhere{\blocklength[\a]}
			\pgfmathsetmacro\offsethere{\offsetsmiddle[\a]}
			\foreach \x in {0,...,\z}
			{

				\draw ($(\theendofthisgrid-\offsethere*\horizontalscale*0.5-\thelengthhere*\horizontalscale*0.5*\x,0.5-0.5*\a)$) -- ($(\theendofthisgrid-\offsethere*\horizontalscale*0.5-\thelengthhere*\horizontalscale*0.5*\x,0.5-0.5*\a-0.5)$);
			}
		}

		\foreach \a/\z in {1/2,2/1,3/1}
		{
			\pgfmathsetmacro\thelengthhere{\blocklength[\a]}
			\pgfmathsetmacro\offsethere{\offsetsmiddle[\a]}
			\foreach \x in {1,...,\z}
			{
				 \draw[red]
				 ($(\theendofthisgrid-\offsethere*\horizontalscale*0.5+\thelengthhere*\horizontalscale*0.5*\x,0.5-0.5*\a)$) -- ($(\theendofthisgrid-\offsethere*\horizontalscale*0.5+\thelengthhere*\horizontalscale*0.5*\x,0.5-0.5*\a-0.5)$);
			}
		}
		\foreach \a/\z in {1/1,2/1,3/1}
		{
			\pgfmathsetmacro\thelengthhere{\blocklength[\a]}
			\foreach \x in {1,...,\z}
			{

				\draw[red] ($(\theendoffourthhalf-\thelengthhere*\horizontalscale*0.5*\x,0.5-0.5*\a)$) -- ($(\theendoffourthhalf-\thelengthhere*\horizontalscale*0.5*\x,0.5-0.5*\a-0.5)$);
			}
		}

		\foreach \x/\y/\theclasshere/\directionhere/\stylehere in {
			1/2/A/left/dotted,1/13/A/left/draw,1/24/A/left/draw,1/35/A/left/draw,1/46/A/left/draw,
			1/6/B/right/dotted,1/17/B/right/draw,1/28/B/right/draw,1/39/B/right/draw,1/50/B/right/draw,
			2/2/C/right/dotted,2/20/C/right/draw,2/38/C/right/draw,
			2/12/E/left/dotted,2/30/E/left/draw,2/48/E/left/draw,
			3/2/F/left/dotted,3/18/F/left/draw,3/34/F/left/draw,3/50/F/left/draw,
			3/9/D/right/dotted,3/25/D/right/draw,3/41/D/right/draw}
		{

			\begin{scope}
				\clip (0,0) rectangle (\theendoffirsthalf,-3);
			\node[gridnode, equivclass\theclasshere] (a) at (\y*\horizontalscale*0.5,-\x*0.5+0.25) {};
			\pgfmathparse{\blocklength[\x]}\def\thelengthhere{\pgfmathresult}
			\node[gridnode, equivclass\theclasshere] (b) at (\y*\horizontalscale*0.5+\thelengthhere*0.5*\horizontalscale,-\x*0.5+0.25) {};
			\path[draw, thick, equivclass\theclasshere, \stylehere]
			(a) to[bend \directionhere, distance = 3mm] (b);

			\end{scope}
		}

		\foreach \x/\y/\theclasshere/\directionhere/\stylehere in {
			1/52/A/left/draw,1/63/A/left/draw,1/74/A/left/dotted,1/85/A/left/dotted,1/96/A/left/dotted,1/107/A/left/dotted,1/118/A/left/dotted,
			1/56/B/right/draw,1/67/B/right/draw,1/78/B/right/dotted,1/89/B/right/dotted,1/100/B/right/dotted,1/111/B/right/dotted,
			2/44/C/right/draw,2/62/C/right/dotted,2/80/C/right/dotted,2/98/C/right/dotted,2/116/C/right/dotted,
			2/54/E/left/draw,2/72/E/left/dotted,2/90/E/left/dotted,2/108/E/left/dotted,
			3/60/F/left/dotted,3/76/F/left/dotted,3/92/F/left/dotted,3/108/F/left/dotted,
			3/41/D/right/draw,3/51/D/right/draw,3/67/D/right/dotted,3/83/D/right/dotted,3/99/D/right/dotted,3/115/D/right/dotted}
		{

			\begin{scope}
				\clip (\thebeginingofsecondhalf,0) rectangle (\theendofthirdhalf,-3);

				\node[gridnode, equivclass\theclasshere] (a) at (\y*\horizontalscale*0.5,-\x*0.5+0.25) {};
				\pgfmathparse{\blocklength[\x]}\def\thelengthhere{\pgfmathresult}
				\node[gridnode, equivclass\theclasshere] (b) at (\y*\horizontalscale*0.5+\thelengthhere*0.5*\horizontalscale,-\x*0.5+0.25) {};
				\path[draw, thick, equivclass\theclasshere, \stylehere]
				(a) to[bend \directionhere, distance = 3mm] (b);

			\end{scope}
		}

		\foreach \x/\y/\theclasshere/\directionhere/\stylehere in {
			1/120/A/left/dotted,1/131/A/left/dotted,
			1/124/B/right/dotted,1/135/B/right/dotted,
			2/117/C/right/dotted,
			2/127/E/left/dotted,
			3/121/F/left/dotted,
			3/128/D/right/dotted}
		{

			\begin{scope}
				\clip (\thebeginoffourthhalf,0) rectangle (\theendoffourthhalf,-3);

				\node[gridnode, equivclass\theclasshere] (a) at (\y*\horizontalscale*0.5,-\x*0.5+0.25) {};
				\pgfmathparse{\blocklength[\x]}\def\thelengthhere{\pgfmathresult}
				\node[gridnode, equivclass\theclasshere] (b) at (\y*\horizontalscale*0.5+\thelengthhere*0.5*\horizontalscale,-\x*0.5+0.25) {};
				\path[draw, thick, equivclass\theclasshere, \stylehere]
				(a) to[bend \directionhere, distance = 3mm] (b);

			\end{scope}
		}

		\arcadedotsnew{\theendoffirsthalf}{\thebeginingofsecondhalf}
		\arcadedotsnew{\theendofthirdhalf}{\thebeginoffourthhalf}

		\draw [decorate, decoration = {brace}] (-0.125,-3*0.5) -- (-0.125,0)node[pos=0.5,black, left=0.1]	{\strut $k$};
		\draw [decorate, decoration = {brace}] (0,0.1) -- (\ponesize*0.5*\horizontalscale,0.1)node[pos=0.5,black, above=0.1,font=\footnotesize, xshift=-2mm] {\strut $f(k)p_1p_2$};
		\draw [decorate, decoration = {brace}] (\ponesize*0.5*\horizontalscale,0.1) -- (\ponesize*0.5*2*\horizontalscale,0.1) node[pos=0.5,black, above=0.1, font=\footnotesize, xshift=2mm] {\strut $f(k)p_1p_2$};
		\draw [decorate, decoration = {brace}] (\pthreesize*0.5*1*\horizontalscale,-3*0.5-0.1) -- (0,-3*0.5-0.1) node[pos=0.5,black, below=0.1, font=\footnotesize] {\strut $f(k)p_kp_1$};
		\draw [decorate, decoration = {brace}] (\pthreesize*0.5*2*\horizontalscale,-3*0.5-0.1) -- (\pthreesize*0.5*1*\horizontalscale,-3*0.5-0.1) node[pos=0.5,black, below=0.1, font=\footnotesize] {\strut $f(k)p_kp_1$};
	\end{tikzpicture}
	\caption{Sketch of the compression $\compEquiv$ and the additional
		equivalence $\compExtEquiv$.
		The figure shows the intervals of the rows of $C$ (drawn in black) and $C^*$ (additional intervals drawn in red as in Figure~\ref{fig:grids})
		that are identified by the equivalences.
		Some $\compExtEquiv$-equivalence classes are shown, each class in its own color.
		Vertices connected by lines are $\compEquiv$-equivalent and $\compExtEquiv$-equivalent.
		Dotted lines indicate $\compExtEquiv$-equivalence only.
		All vertices in gray areas of $C$ form singleton $\compEquiv$-equivalence classes.}
	\label{fig:grid-equivalence-relation}
\end{figure}
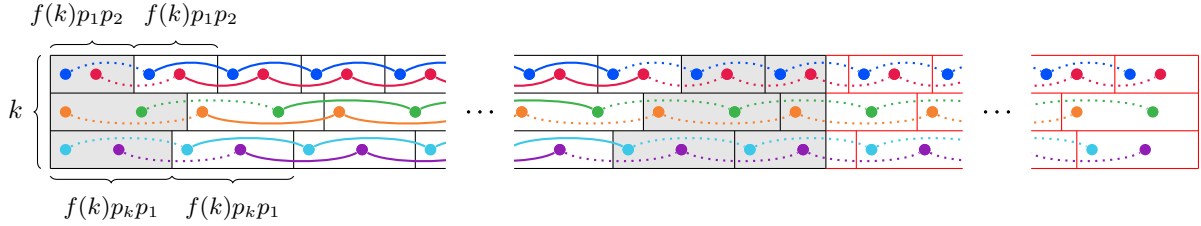

We start by defining an equivalence relation $\compExtEquiv$ on $V(C^*)$.
Vertices are only equivalent to vertices in the same row.
In row~$i$, the distance between equivalent vertices is a multiple of~$f(k) \cdot p_i \cdot p_{i+1}$, where indices are taken modulo~$k$.
More formally, for all~$i,i' \in I$ and~$j, j' \in J^*$, we set
\[(i,j) \compExtEquiv (i',j') \quad \Longleftrightarrow \quad  i = i' \text{ and } (j'-j) \text{ is divisible by } f(k) \cdot p_i \cdot p_{i+1}.\]
From the equivalence~$\compExtEquiv$ we obtain the equivalence $\compEquiv$ on $V(C)$, which is used for the compression (note that $V(C) \subseteq V(C^*)$), as follows.
In row $i$, we use the equivalence $\compExtEquiv$ except that several vertices are made singleton equivalence classes.
Specifically, all vertices in the complete first interval of length $f(k) \cdot p_i \cdot p_{i+1}$ are singletons. Also vertices in the last uncropped repetition of this~$f(k) \cdot p_i \cdot p_{i+1}$ length interval, plus possibly the vertices in a subsequent cropped interval, are singletons.
More formally, for all~$i,i' \in I$ and~$j \leq j' \in J$, we set
\[(i,j) \compEquiv (i',j') \quad \Longleftrightarrow \quad (i,i') \compExtEquiv (i',j') \text { and } j\geq f(k) \cdot p_i \cdot p_{i+1}  \text { and } j' < \lambda_i \cdot f(k) \cdot p_i \cdot p_{i+1},\]
where $\lambda_i$ is the largest integer such that $(\lambda_i+1) \cdot f(k) \cdot p_i \cdot p_{i+1} \leq |J|$.
A visualization is given in Figure~\ref{fig:grid-equivalence-relation}.
In particular, note that all vertices in the first and last~$f(k)$ columns of~$C$ form singleton $\compEquiv$-equivalence classes.

\begin{lemma}
	\label{lem:grid-width}
	The cylindrical grid $C$ has width $(w/2)^k \leq |J| \leq (k+1) \cdot w^{k}$
	and~$\compEquiv$ has $\Theta(k^2 \cdot w^2)$ many equivalence classes.
\end{lemma}

\begin{proof}
	We have $|J| = \frac{1}{2} \cdot f(k) \cdot p_1 \cdot p_2\cdots p_k$.
	Since $w/2 \leq p_i \leq w$ for every $i \in I$, it follows that $(w/2)^k \leq |J| \leq (k+1) \cdot w^{k}$.

	To bound the number of equivalence classes, consider an arbitrary row $i \in I$.
	Then there are at least $3 \cdot f(k) \cdot p_i \cdot p_{i+1}$
	and at most $4 \cdot f(k) \cdot p_i \cdot p_{i+1}$ equivalence classes in the $i$-th row.
	Using again that $w/2 \leq p_i \leq w$ for every $i \in I$, we obtain $\Theta(k^2 \cdot w^2)$ equivalence classes in total.
\end{proof}

\begin{lemma}
	\label{lem:relation-is-compression}
	The equivalence relation~$\compEquiv$ is a $C$-compression.
\end{lemma}

\begin{proof}
	Let $\vertA,\vertB \in V(C)$ such that $\vertA \compEquiv \vertB$ and $\vertA \neq \vertB$.
	We first show that~$\vertA$ and~$\vertB$ are not adjacent and have the same degree.
	By construction,~$\vertA$ and~$\vertB$ have distance at least~$f(k)$, which in particular implies that they are not adjacent.
	Also, $\vertA$ and~$\vertB$ are not in first or last column which implies that both have degree~$4$.

	Next, let $\set{\vertA,\vertB},\set{\vertA',\vertB'}\in E(C)$ such that $\vertA\compEquiv\vertA'$ and $\vertB\compEquiv\vertB'$.
	Then the two edges are either both horizontal edges in the same row or both vertical edges between the same rows.
	Also, $\vertA$ and $\vertA'$ are in the same row and $\vertB$ and $\vertB'$ are in the same row.
	Because the cylindrical grid is ordered lexicographically,
	if $\vertB$ is the $i$-th neighbor of $\vertA$
	then $\vertB'$ is the $i$-th neighbor of $\vertA'$.
\end{proof}

Now that we established that $\compEquiv$ is a compression, we can use all results from Section~\ref{sec:comp-cfi}.
The goal is to prove an $\Omega(w^k)$ lower bound on the iteration number of $k$-WL on compressed CFI graphs over~$C$.
By Lemma \ref{lem:compressed-cops-and-robber},
it suffices to show that the robber has a winning strategy in $\Omega(w^k)$ rounds in the compressed $(k+1)$-Cops and Robber game played on $(C,\compEquiv)$.
Towards this, we define twistings to move the robber.

\begin{definition}[End-to-End Twisting]
	For a set of vertices~$W \subseteq V(C)$, an \defining{end-to-end twisting avoiding~$W$} is a $\compEquiv$-compressible $C$-twisting that
	\begin{enumerate}
		\item twists some edge $e_1$ in the first  column and some edge $e_2$ in the last column of $C$,
		\item twists no other edge, and
		\item fixes every vertex in $\eqclass{\compEquiv}{W}$.
	\end{enumerate}
\end{definition}

We now consider situations in which there are end-to-end twistings avoiding a set~$W\subseteq V(C)$.
For~$\ell \geq 2$, we define the equivalence relation $\approx_\ell$ on $V(C^*)=I \times J^*$ in such a way that,
for all $i,i' \in I$ and $j \leq j' \in J^*$, we have
\[(i,j) \approx_\ell (i',j') \quad \Longleftrightarrow \quad i=i' \text{ and } (j'-j) \text{ is divisible by } \ell.\]
Note that we use the same period in every row for~$\approx_\ell$.
Slightly abusing notation, we also use~$\approx_\ell$ to denote the corresponding equivalence relation on $V(C)$.

\begin{definition}[Pairwise-Separator]
	For a pair of consecutive rows~$i$ and~$i+1$ (indices taken modulo~$k$),
	a set $W \subseteq V(C)$ is a \defining{pairwise-separator for rows $i$ and $i+1$}
	if $\eqclass{\approx_{f(k) p_{i+1}}}{W}$ separates the first and last column in the subgraph of $C$ induced by rows~$i$ and~$i+1$.
	The set $W$ is a \defining{pairwise-separator} if $W$ is a pairwise-separator for rows $i$ and $i+1$ for every $i \in I$.
\end{definition}

A path $P = (\vertA_1, \dots, \vertA_m)$ in $C$ is \defining{$\ell$\nobreakdash-periodic} if
$\vertA_1$ and $\vertA_m$ are in singleton $\compEquiv$\nobreakdash-classes and,
for every $i < m$
and every $\vertB \in V(C)$ such that $\vertB\approx_\ell \vertA_i$
and $\eqclass{\compEquiv}{\vertA_i}$ and $\eqclass{\compEquiv}{\vertB}$
are not singleton classes,
there is a $j < m$ such that
$\vertA_j = \vertB$ and
$\vertA_{i+1} \approx_\ell \vertA_{j+1}$.
The path induces the $C$-twisting
\[T_P \coloneqq \setcond[\big]{(\vertA_i,\vertA_{i-1}),(\vertA_i,\vertA_{i+1})}{1 < i < m}.\]

\begin{lemma}
	\label{lem:periodic-implies-compressible}
	Let $I' \subseteq I$ be a set of rows
	and $P = (\vertA_1, \dots, \vertA_m)$ be a path in $C$
	that only uses vertices from rows in~$I'$.
	Let~$q$ be the greatest common divisor of all~$f(k)p_ip_{i+1}$ for $i \in I'$
	(that is, $q = f(k)p_ip_{i+1}$ if $I' = \set{i}$,
	$q = f(k)p_{i+1}$ if $I'=\set{i,i+1}$, and $q=f(k)$ otherwise).
	If~$P$ is $q$\nobreakdash-periodic,
	then the induced $C$-twisting~$T_P$
	is $\compEquiv$-compressible,
	twists the edges $\set{\vertA_1,\vertA_2}$ and $\set{\vertA_{m-1},\vertA_m}$,
	twists no other edges, and
	fixes all vertices apart from~$u_1,\dots,u_m$.
\end{lemma}

\begin{proof}
	The $C$-twisting $T_P$ twists exactly the edges $\set{\vertA_1,\vertA_2}$ and $\set{\vertA_{m-1}, \vertA_m}$ and fixes all vertices apart from $u_1,\dots,u_m$ by construction.
	To show that~$T_P$ is $\compEquiv$-compressible,
	we have to show that, for all $\vertA \compEquiv \vertA'$ of degree $d$,
	we have $(\vertA,\vertB_\ell) \in T_P$ if and only if $(\vertA',\vertB'_\ell) \in T_P$ for every $\ell \in [d]$,
	where~$\vertB_\ell$ respectively~$\vertB'_\ell$
	is the $\ell$-th neighbor of~$\vertA$ respectively~$\vertA'$.
	It suffices to show one direction of the equivalence,
	the other one follows by swapping $\vertA$ and $\vertA'$.
    So let $\vertA \compEquiv \vertA'$ and~$\vertB_\ell$ respectively~$\vertB'_\ell$
    be the $\ell$-th neighbor of~$\vertA$ respectively~$\vertA'$
    for an arbitrary~$\ell$.
    Assume $(\vertA,\vertB_\ell) \in T_P$.
    We show $(\vertA',\vertB'_\ell) \in T_P$.
	If~$\eqclass{\compEquiv}{\vertA}$ is a singleton class,
	then $\vertA=\vertA'$ and $\vertB_\ell =\vertB'_\ell$,
	and we are done.
	So suppose that~$\eqclass{\compEquiv}{\vertA}$ is not a singleton class.
	In particular,~$\vertA$ and~$\vertA'$ are neither the first nor last vertex in~$P$ because $P$ is $q$-periodic.
	Let $\vertA=(s,t)$ and $\vertA' =(s', t')$.
	We have $s = s'$ because $\vertA \compEquiv \vertA'$.
	Because~$f(k)p_sp_{s+1}$ is divisible by~$q$,
	we also have $\vertA \approx_q \vertA'$.
	Recall that we ordered the vertices lexicographically.
	That is, if~$\vertB_\ell$ is the right neighbor of~$\vertA$,
	then~$\vertB'_\ell$ is the right neighbor of~$\vertA'$ and similar for all other directions.
	Hence, $\vertB_\ell \approx_q \vertB'_\ell$.
	Because $(\vertA,\vertB_\ell) \in T_P$,
	there is an~$i$ such that $\set{\vertA,\vertB_\ell} = \set{\vertA_i, \vertA_{i+1}}$.
	Because~$\vertA$ is neither the first nor the last vertex in~$P$,
	we may assume without loss of generality that $\vertA = \vertA_i$ and $\vertB_\ell = \vertA_{i+1}$.
	Since~$P$ is $q$\nobreakdash-periodic,
	there is a~$j$ such that $\vertA' = \vertA_j$ and $\vertB'_\ell = \vertA_{j+1}$.
	This implies that $(\vertA_j,\vertA_{j+1}) = (\vertA',\vertB'_\ell) \in T_P$.
\end{proof}

We use periodic paths in Lemmas~\ref{lem:no-pair-sep-twisting} and~\ref{lem:no-pseudo-sep-twisting} to construct end-to-end twistings.

\begin{lemma}
	\label{lem:no-pair-sep-twisting}
	Let $W \subseteq V(C)$ be a set of at most $k$ vertices of $C$ such that $W$ is not a pairwise-separator.
	Then there is an end-to-end twisting avoiding~$W$.
\end{lemma}

\begin{figure}
	\centering
	\begin{tikzpicture}[font=\small, scale=0.85]
		\node at (-0.5,0.7) {$i$};
		\node at (-0.5,0) {$i+1$};

		\foreach \i in {1,...,20}{
			\node[gridnode] (v\i) at (0.7*\i,0) {};
			\node[gridnode] (w\i) at (0.7*\i,0.7) {};
		}

		\foreach \i in {1,...,20}{
			\draw (v\i) -- (w\i);
			\draw (v\i) -- ($(v\i.center)-(0,0.3)$);
			\draw (w\i) -- ($(w\i.center)+(0,0.3)$);
		}

		\foreach \i/\j in {1/2,2/3,3/4,4/5,5/6,6/7,7/8,8/9,9/10,10/11,11/12,12/13,13/14,14/15,15/16,16/17,17/18,18/19,19/20}{
			\draw (v\i) -- (v\j);
			\draw (w\i) -- (w\j);
		}

		\draw (v1) -- ($(v1.center)-(0.3,0)$);
		\draw (w1) -- ($(w1.center)-(0.3,0)$);
		\draw (v20) -- ($(v20.center)+(0.3,0)$);
		\draw (w20) -- ($(w20.center)+(0.3,0)$);

		\node[rotate = -90] at (2.8,-0.56) {$j$};
		\node[rotate = -90] at (7.0,-0.9) {$j + q$};
		\node[rotate = -90] at (11.2,-1.0) {$j + 2q$};

		\node at (0.3,0.35) {$\dots$};
		\node at (14.4,0.35) {$\dots$};

		\begin{pgfonlayer}{background}
			\foreach \i in {3,5,9,11,15,17}{
				\draw[line width = 4pt, color = red!60] (v\i.center) -- (w\i.center);
			}

			\foreach \i/\j in {1/2,2/3,5/6,6/7,7/8,8/9,11/12,12/13,13/14,14/15,17/18,18/19,19/20}{
				\draw[line width = 4pt, color = red!60] (w\i.center) -- (w\j.center);
			}

			\foreach \i/\j in {3/4,4/5,9/10,10/11,15/16,16/17}{
				\draw[line width = 4pt, color = red!60] (v\i.center) -- (v\j.center);
			}
			\draw[line width = 4pt, color = red!60] (w1.center) -- ($(w1.center)-(0.3,0)$);
			\draw[line width = 4pt, color = red!60] (w20.center) -- ($(w20.center)+(0.3,0)$);
		\end{pgfonlayer}
	\end{tikzpicture}
	\caption{The figure shows the $q$-periodic path (highlighted in red) for $q = f(k)p_{i+1}$ constructed in the proof of Lemma~\ref{lem:no-pair-sep-twisting}.}
	\label{fig:periodic-path}
\end{figure}
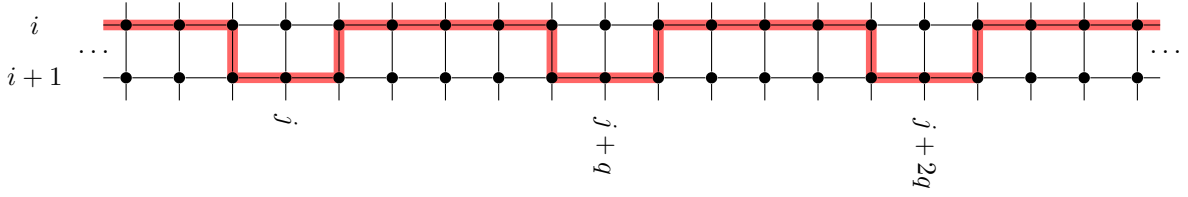

\begin{proof}
	Suppose~$W$ is not a pairwise-separator. If~$W$ does not contain a vertex from every row,
	then there is a row $i$ such that no vertices of row $i$ are in $W$.
	Hence, there is a path from one end to the other that avoids $W$ and only uses vertices of row $i$.
	This path is trivially $(f(k)p_ip_{i+1})$-periodic and thus induces the desired end-to-end twisting by Lemma~\ref{lem:periodic-implies-compressible}.

	So suppose otherwise.
	There is some $i \in I$ such that~$W$ is not a pairwise-separator for rows~$i$ and~$i+1$.
	The set~$W$ contains exactly one element from row~$i$, say~$(i,j)$, and one from row~$i+1$, say~$(i+1,j')$.
	Since $\eqclass{\approx_{f(k) p_{i+1}}}{W}$ does not separate the ends of rows~$i$ and~$i+1$, we conclude that~$(j-j')\bmod f(k) p_{i+1}\notin \set{-1,0,1}$.
	Then the path that contains the vertices
	\begin{align*}
		&\setcond[\big]{(i,\ell) }{ (j-\ell) \bmod f(k) p_{i+1} \neq 0} \text{ and}\\
		&\setcond[\big]{(i+1,\ell) }{ (j-\ell) \bmod f(k) p_{i+1} \in \set{-1,0,1}}
	\end{align*}
	is $(f(k) p_{i+1})$\nobreakdash-periodic (see Figure \ref{fig:periodic-path})
	and induces the desired end-to-end twisting by Lemma~\ref{lem:periodic-implies-compressible}.
\end{proof}

In the cylindrical grid $C$, we are interested in $k$-vertex sets $W\subseteq V(C)$
such that $\eqclass{\compEquiv}{W}$ separates the first from the last column.
To analyze how such separators can be moved,
we consider the toroidal grid $C^*$.
We say that the $j$-th and $j'$-th column ($j,j' \in J^*$) of $C^*$ are \defining{consecutive} if $(j-j') \bmod |J^*| \in \set{-1,0,1}$.

\begin{definition}[Vertical Separator]
	A set $W\subseteq V(C)$ is a \defining{vertical separator}
	if in the graph $C-W$ the first column is separated from the last column.
	A set $W \subseteq V(C^*)$ is a \defining{toroidal vertical separator}
	if there is a number $z \in \ZZ$ such that
	$W$ shifted by $z$ columns in $C^*$, i.e., the set $\setcond{(i,(j+z) \bmod |J^*|)}{(i,j)\in W} \cap V(C)$, is a vertical separator of $C$.
\end{definition}

In particular, every vertical separator is a toroidal vertical separator.
We regularly rely on the following simple properties of $k$-vertex toroidal vertical separators.

\begin{lemma}
	\label{lem:vertical-separator-k-columns}
	Let $S \subseteq V(C^*)$ be a $k$-vertex toroidal vertical separator.
	\begin{enumerate}[label = (\roman*)]
		\item\label{item:vertical-separator-k-columns-1} $S$ contains exactly one vertex per row.
		\item\label{item:vertical-separator-k-columns-2} If $(i,j)$ and $(i+1,j')$ are the vertices in row $i$ and $i+1$ in $S$ (row indices modulo $k$), then the $j$-th and $j'$-th column are consecutive in~$C^*$.
		\item\label{item:vertical-separator-k-columns-3} $S$ spans at most $k$ consecutive  columns of $C^*$,
		i.e., there is column $j\in J^*$ such that $S \subseteq \setcond{(i,(j+\ell)\bmod |J^*|)}{i, \ell \in [0,k-1]}$.
	\end{enumerate}
\end{lemma}

\begin{proof}
	Because there are $k$ rows and $S$ contains $k$ vertices, Part~\ref{item:vertical-separator-k-columns-1} follows.
	For Part~\ref{item:vertical-separator-k-columns-2}, suppose $(i,j), (i+1,j') \in S$ are the unique vertices in rows $i$ and $i+1$ in $S$.
	Since $S$ is a toroidal vertical separator, there is some $z \in \ZZ$ such that $S_z \coloneqq \setcond{(i,(j+z) \bmod |J^*|)}{(i,j)\in W} \cap V(C)$ is a vertical separator of $C$.
	Suppose towards a contradiction that the $j$-th and $j'$-th columns are not consecutive.
	If one of $(i, (j + z) \bmod |J^*|)$ and $(i',(j'+z) \bmod |J^*|)$ is not contained in $V(C)$,
	then there is clearly a path from the first to the last column in $C - S_z$.
	Otherwise, assume by symmetry $(j + z) \bmod |J^*| \leq (j'+z) \bmod |J^*|$.
	Because the $j$-th and $j'$-th column are not consecutive, we have $j + 1 < j'$.
	Hence, $((i+1,0), \dots,(i+1,j+1), (i,j+1),\dots, (i,|J|-1))$ is a path in $C-S_z$.
	In both cases, we obtain a contradiction.
	Finally, Part~\ref{item:vertical-separator-k-columns-3} immediately follows from Part~\ref{item:vertical-separator-k-columns-2}.
\end{proof}

We define~${\approx} \coloneqq {\approx_{f(k)}}$ (as before, slightly abusing notation, we use~${\approx}$ on both graphs~$C$ and~$C^*$).

\begin{definition}[Pseudo-Separator]
	A set $W \subseteq V(C)$ is a \defining{pseudo-separator}
	if $W$ is a pairwise-separator and $\pseudoclass{W}$ is a vertical separator.
\end{definition}

\begin{lemma}
	\label{lem:no-pseudo-sep-twisting}
	Let $W \subseteq V(C)$ be a set of at most $k$ vertices such that $W$ is not a pseudo-separator.
	Then there is an end-to-end twisting avoiding~$W$.
\end{lemma}

\begin{proof}
	If $W$ is not a pairwise-separator,
	then there is an end-to-end twisting avoiding~$W$ by Lemma~\ref{lem:no-pair-sep-twisting}.

	Otherwise, $\pseudoclass{W}$ is not a vertical separator.
	In this case, there is a path~$P$ from a vertex in column~$0$ to a vertex in column $2f(k)-1$ that avoids~$\pseudoclass{W}$.
	Since~$W$ has at most~$k$ elements and~$f(k)>k$, there is a column~$j \in [0,f(k)-1]$ that contains no element of~$\pseudoclass{W}$.
	Let~$P'$ be a subpath of~$P$ containing exactly one vertex from column~$j$,
	one vertex from column~$j+f(k)$ and otherwise only vertices in the columns between~$j$ and~$j+f(k)$.
	Such a subpath exists because~$P$ starts in column~$0$ and ends in column~$2f(k)-1$.

	Suppose~$(i,j)$ is the one endpoint of~$P'$ and~$(i',j+f(k))$ is the other.
	We extend~$P'$ to a path~$Q$ by adding as initial segment a path from~$(i',j)$ to~$(i,j)$ in column $j$ (if~$i=i'$ then~$Q=P'$).
	Note that~$Q$ is a path spanning~$f(k)$ columns which starts and ends in the same row.
	Also note that~$Q$ avoids~$\pseudoclass{W}$.

	Finally, we construct the path~$\widehat{Q}$ by shifting~$Q$ by multiples of~$f(k)$,
	that is, if $Q$ uses vertices $V(Q)$ and edges $E(Q)$,
	then~$\widehat{Q}$ has vertices
	\[V(\widehat{Q}) \coloneqq \setcond[\big]{(i,j + z \cdot f(k))}{(i,j) \in V(Q), z\in \mathbb{Z}}\cap V(C)\]
	and edges
	\[E(\widehat{Q}) \coloneqq \setcond[\Big]{\set[\big]{(i,j + z \cdot f(k)),(i',j' + z \cdot f(k))}}{\set[\big]{(i,j),(i',j')} \in E(Q), z\in \mathbb{Z}}\cap E(C).\]
	The path $\widehat{Q}$ is $f(k)$-periodic by construction and induces the desired end-to-end twisting avoiding~$W$ by Lemma~\ref{lem:periodic-implies-compressible}
	(note that $\widehat{Q}$ is trivially $(f(k)p)$-periodic for every positive $p \in \nat$).
\end{proof}

If $W$ is a pseudo-separator, we cannot construct end-to-end twistings avoiding~$W$ in general.
In this case, the goal is to show that pseudo-separators not containing proper vertical separators cannot be used to catch the robber.

\begin{lemma}
	\label{lem:only-one-separator}
	Let $W \subseteq V(C)$ be a set of at most~$k$ vertices.
	\begin{enumerate}
		\item There is at most one $k$-vertex toroidal
                  vertical separator $S_W \subseteq
                  \eqclass{\compExtEquiv}{W}$
		(and thus at most one vertical separator).
		\item If $\eqclass{\compExtEquiv}{W}$ is a toroidal vertical separator,
		then there is a $k$-vertex toroidal vertical separator $S_W \subseteq \eqclass{\compExtEquiv}{W}$.
	\end{enumerate}
\end{lemma}

\begin{proof}
	We prove the first claim.
	If~$\eqclass{\compExtEquiv}{W}$ contains some $k$-vertex toroidal vertical separator~$S_W$,
	then~$W$ contains exactly one vertex from every row using Lemma \ref{lem:vertical-separator-k-columns}\ref{item:vertical-separator-k-columns-1}. Suppose~$S'_W\subseteq \eqclass{\compExtEquiv}{W}$ is a second $k$-vertex toroidal vertical separator.
	Note that~$S_W$ and~$S'_W$ also contain exactly one vertex from every row.  Let us make three observations:

	\begin{enumerate}[label = (\alph*)]
		\item\label{item:only-one-separator-1}
		$S_W$ is contained in $k$ consecutive columns of $C^*$ by Lemma~\ref{lem:vertical-separator-k-columns}\ref{item:vertical-separator-k-columns-3}.
		So, if~$(s,t),(i,j)\in S_W$, then~$((t-j) \bmod |J^*|) < k$.
		\item\label{item:only-one-separator-2} Similarly, if~$(s,t'),(i,j')\in S'_W$, then~$((t'-j') \bmod |J^*|) < k$.
		\item\label{item:only-one-separator-3} Suppose~$(i,j)$ and~$(i,j')$ are vertices in row~$i$ in~$S_W$ and~$S'_W$, respectively.
			Both vertices are in particular in~$\eqclass{\compExtEquiv}{W}$.
			Thus,~$(j-j') \bmod f(k) p_{i} p_{i+1} = 0 $.
	\end{enumerate}
	We claim that, for all $i \in I$, all~$(s,t) \in S_W$, and all $(s,t') \in S_W'$, we have
	\begin{equation}
		\label{eq:only-one-separator-1}
		(t-t') \bmod f(k) p_{i} p_{i+1} \in \{-2k,\ldots,2k\}.
	\end{equation}
	To see this, let $(i,t_i)$ be the vertex of row~$i$ in~$S_W$, and let $(i,t_i')$ be the vertex of row~$i$ in~$S_W'$.
	By \ref{item:only-one-separator-3}, we get $(t_i-t_i') \bmod f(k)p_ip_{i+1} = 0$, and by \ref{item:only-one-separator-1} and \ref{item:only-one-separator-2}, we get $((t-t_i) \bmod |J^*|) < k$ and $((t'-t_i') \bmod |J^*|) < k$, respectively.
	This implies the claim since $f(k)p_ip_{i+1}$ divides $|J^*|$.

	Furthermore, by \ref{item:only-one-separator-3} (using~$s$ in place of~$i$) we get that
	\[(t-t') \bmod f(k) p_{s} p_{s+1} = 0\]
	which in particular implies that
	\begin{equation}
		\label{eq:only-one-separator-2}
		(t-t') \bmod f(k) = 0.
              \end{equation}
	Since~$f(k)> 2k$, we can combine Equations~\eqref{eq:only-one-separator-1} and~\eqref{eq:only-one-separator-2} and obtain that
	\[(t-t') \bmod f(k) p_i p_{i+1} = 0\]
	for all $i \in [k]$.
	It follows that
	\[(t-t') \bmod  f(k) p_{1} p_{2}\cdots p_{k} = 0\]
	since all $p_i$ are pairwise coprime.
	This implies $t=t'$.
	So overall, for all~$(s,t)\in S_W$ and all~$(s,t')\in S'_W$, we have~$t=t'$.
	This means that~$S_W = S_W'$.

	It remains to show the second claim.
	Suppose that $\eqclass{\compExtEquiv}{W}$ is a toroidal vertical separator.
	Note that $\eqclass{\compExtEquiv}{W}$ contains a vertex from every row.
	Because~$W$ is of size at most~$k$,
	the set $\eqclass{\compExtEquiv}{W}$ consists of exactly one $\compExtEquiv$-equivalence class per row.
	Let~$S_W$ be a minimal toroidal vertical separator such that $S_W\subseteq \eqclass{\compExtEquiv}{W}$.
	For the sake of contradiction suppose that $|S_W| > k$.
	Assume that~$S_W$ contains two vertices $(i,j)$ and $(i,j')$ of row~$i$.
	Consider the columns $j, \dots, j + f(k)-1$ (column indices modulo $|J^*|$).
	Then $S_W$ contains at most $k$ vertices of these columns
	because $\compExtEquiv$-equivalent vertices in the same row have distance at least $f(k)$ and  $\eqclass{\compExtEquiv}{W}$ contains
	at most one $\compEquiv$-equivalence class per row.
	That is, there is a column between the $j$-th and the $(j+f(k))$-th column
	of which no vertex is contained in~$S_W$.
	Similarly, there is also a column between the $j'$-th and the $(j'+f(k))$-th column of which no vertex is contained in~$S_W$.
	So, for both horizontal directions between columns $j$ and $j'$ in the toroidal grid (one towards to the next larger column index and the other one towards the next smaller column index modulo $|J^*|$),
	there is column  of which no vertex is contained in~$S_W$.
	Hence, at least one of $(i,j)$ and $(i,j')$ can be removed from~$S_W$
	such that we still obtain a vertical toroidal separator.
	This contradicts~$S_W$ being minimal.
\end{proof}

Using these unique $k$-vertex toroidal vertical separators, we show that a pseudo-separator~$W$ that is not a toroidal vertical separator cannot be turned into a toroidal vertical separator by exchanging a single vertex in~$W$.

\begin{lemma}
	\label{lem:pseudo-adjacent-to-sep-is-sep}
	Let~$W, W' \subseteq V(C)$ be pseudo-separators each of size~$k$ differing by at most one vertex~(i.e.,~$|W\cap W'|\geq k-1$).
	Suppose that $\eqclass{\compExtEquiv}{W}$ is a toroidal vertical separator. Then
	\begin{enumerate}
        \item\label{item:pseudo-adjacent-to-sep-is-sep-1}
        there is a $k$-vertex toroidal separator $S_{W'} \subseteq \eqclass{\compExtEquiv}{W'}$,
		\item\label{item:pseudo-adjacent-to-sep-is-sep-2} the unique $k$-vertex toroidal vertical separators~$S_W \subseteq \eqclass{\compExtEquiv}{W}$ and~$S_{W'} \subseteq \eqclass{\compExtEquiv}{W'}$ (see Lemma~\ref{lem:only-one-separator}) differ by at most one vertex~(i.e.,~$|S_W\cap S_{W'}|\geq k-1$), and
		\item\label{item:pseudo-adjacent-to-sep-is-sep-3} every $\vertA \in S_{W'}$ has distance at most~$2$ to~$S_W$ in~$C^*$.
	\end{enumerate}
\end{lemma}

\begin{proof}
	Note that~$W$ and~$W'$ each contain exactly one vertex from every row.
	If~$\eqclass{\compExtEquiv}{W} = \eqclass{\compExtEquiv}{W'}$, there is nothing to show, because the set $S_W$ only depends on $\eqclass{\compExtEquiv}{W}$ (Lemma~\ref{lem:only-one-separator}).

	So assume that~$\eqclass{\compExtEquiv}{W} \neq \eqclass{\compExtEquiv}{W'}$ and let $i \in I$ denote the unique row where~$W$ and~$W'$ differ.
	By Lemma~\ref{lem:only-one-separator}, there is a unique $k$-vertex toroidal vertical separator $S_W \subseteq \eqclass{\compExtEquiv}{W}$.
	Let~$v_i = (i,j_i)$ denote the vertex of~$S_W$ in row~$i$.

	\begin{claim}
		\label{claim:find-replacement}
		There is a vertex~$v_i^*= (i,j_i^*) \in \pseudoclass{W'}$ such that~$(S_W\setminus \{v_i\})\cup \{v_i^*\}$ is a toroidal vertical separator.
	\end{claim}
	\begin{claimproof}
		Since~$W'$ is a pseudo-separator, the set~$\pseudoclass{W'}$ contains some $k$-vertex toroidal vertical separator~$S'$.
		The elements in~$S'$ span at most~$k$ consecutive columns by Lemma~\ref{lem:vertical-separator-k-columns}\ref{item:vertical-separator-k-columns-3}.
		This implies that the set~$\pseudoclass{W'}$ is a periodic set of $k$-vertex toroidal vertical separators that repeat every~$f(k)$ columns.
		More precisely, for every $z \in \mathbb{Z}$, we define
		\[S_z' \coloneqq \setcond[\big]{(i,(j + z \cdot f(k)) \bmod |J^*|)}{(i,j) \in S'}.\]
		Every $S_z'$ is a $k$-vertex toroidal vertical separator because we just shifted $S'$ by $z \cdot f(k)$.

		When regarding the columns of elements from different toroidal vertical separators, the distance is at least ${f(k)-k}$.
		Conversely, elements of~$\pseudoclass{W'}$ whose columns are less than~${f(k)-k}$ apart are in the same toroidal vertical separator.
		Thus, there exists an index $z \in \ZZ$ such that~${S_W\setminus \{v_i\} \subseteq S_z'}$.
		Hence, there is a vertex~$v_i^*\coloneqq (i,j_i^*) \in \pseudoclass{W'}$ in row~$i$ so that the set~$(S_W\setminus \{v_i\})\cup \{v_i^*\} = S'_z$ is a toroidal vertical separator.
	\end{claimproof}

	Let $v_i^* = (i,j_i^*) \in \pseudoclass{W'}$ be the vertex from Claim \ref{claim:find-replacement} such that
	$S_{W'} \coloneqq (S_W \setminus \{v_i\})\cup \{v_i^*\}$ is a toroidal vertical separator.
	In particular,~$v_i^*$ has distance at most~$2$ to~$v_i$ in~$C^*$ by Lemma~\ref{lem:vertical-separator-k-columns}\ref{item:vertical-separator-k-columns-2}.

	\begin{claim}
		\label{claim:in-ext-eq-class}
		$S_{W'} \subseteq \eqclass{\compExtEquiv}{W'}$.
	\end{claim}
	\begin{claimproof}
		First observe that $S_W \setminus \{v_i\} \subseteq \eqclass{\compExtEquiv}{W'}$ since $W$ and $W'$ only differ in row $i$.
		So it remains to argue that~$v_i^*\in \eqclass{\compExtEquiv}{W'}$.
		Let~$v_i' = (i,j_i')$ denote the unique vertex from~$W'$ in row~$i$.

		The sets $S_W$ and $(S_W\setminus \{v_i\})\cup \{v_i'\}$ contain the same vertex $(i+1,j_{i+1})$ in row $(i+1)$ (row indices are modulo~$k$).
		Since $W,S_W \subseteq \eqclass{\compExtEquiv}{W}$ are both $k$-element sets containing one~vertex from each row, we conclude that $\eqclass{\compExtEquiv}{W} = \eqclass{\compExtEquiv}{S_W}$.
		Similarly, $\eqclass{\compExtEquiv}{W'} = \eqclass{\compExtEquiv}{((S_W\setminus \{v_i\})\cup \{v_i'\})}$.
		Together with the fact that $W$ and $W'$ are pseudo-separators (and in particular pairwise-separators), this implies
		\[(j_i-j_i') \bmod f(k) p_{i+1} \in \{-2,\ldots,2\}.\]
		Similarly, $S_W$ and $(S_W\setminus \{v_i\})\cup \{v_i'\}$ contain the same element~$(i-1,j_{i-1})$ in row~$(i-1)$.
		Thus,
		\[(j_i-j_i') \bmod f(k) p_i \in \{-2,\ldots,2\}.\]
		In fact, we have more strongly that
		\[(j_i-j_i') \bmod f(k) p_{i+1} = (j_i-j_i') \bmod f(k) p_i.\]
		Together these statements imply that
		\[(j_i-j_i') \bmod f(k)p_i p_{i+1} \in \{-2,\ldots,2\}\]
		by the Chinese Remainder Theorem.
		So there is some vertex~$v_i''=(i,j_i'')$ in row~$i$ in~$\eqclass{\compExtEquiv}{W'}$ with~$((j_i''-j_i) \bmod |J^*|) \in \{-2,\ldots,2\}$.

		Now~$v_i^* = (i,j_i^*)$ and~$v_i''=(i,j_i'')$ are both row~$i$ vertices in~$\pseudoclass{W'}$ in a column at most~$k$ away from~$v_i = (i,j_i)$ in~$C^*$.
		This means~$v_i^*= v_i'' \in \eqclass{\compExtEquiv}{W'}$ since distinct vertices in~$\pseudoclass{W'}$ that are in same row have distance at least~$f(k)$.
	\end{claimproof}

	Hence,~$S_{W'} = (S_W \setminus \{v_i\})\cup \{v_i^*\} \subseteq \eqclass{\compExtEquiv}{W'}$ is a toroidal vertical separator by Claims~\ref{claim:find-replacement} and~\ref{claim:in-ext-eq-class}.
	Together with Lemma \ref{lem:only-one-separator}, this proves Parts~\ref{item:pseudo-adjacent-to-sep-is-sep-1} and~\ref{item:pseudo-adjacent-to-sep-is-sep-2} of the lemma.
	To prove Part~\ref{item:pseudo-adjacent-to-sep-is-sep-3}, let $u \in S_{W'}$.
	If $u \in S_W$, then~$u$ has distance~$0$ to~$S_W$.
	Otherwise, $u = v_i^*$ and $v_i^*$ has distance at most~$2$ to~$v_i \in S_W$ in~$C^*$.
\end{proof}

We are finally ready to consider the compressed Cops and Robber game on the cylindrical grid~$C$ and the compression $\compEquiv$ (see Lemma~\ref{lem:relation-is-compression}).

\begin{lemma}
	\label{ref:lem-robber-lower-bound:new:construction}
	The robber has a winning strategy in the $\Omega(|J|)$-round
	compressed $(k+1)$-Cops and Robber game played on $C$ and $\compEquiv$ when the robber is initially placed on an edge in the first column.
\end{lemma}
\begin{proof}
	We set $\ell \coloneqq |J|/6 - (k+2)$.
	Since $k \geq 3$ we conclude that $(k+2) \leq |J|/15$ and thus, $\ell = \Omega(|J|)$.
	We now show that the robber has a strategy to win the $\ell$-round game.

	In the game, the cops can be placed onto at most $k+1$ many $\compEquiv$-equivalence classes.
	Instead of considering sets of $\compEquiv$-equivalence classes,
	we will consider sets of representatives of these classes:
	A set $X \subseteq V(C)$ of size at most $k+1$ is called a \defining{cop position}.
	Similarly, a set $W \subseteq V(C)$ of size at most $k$ is called an
	\defining{intermediate cop position}.
	With these notions,
	the game is played as follows:
	\begin{enumerate}
		\item In the current cop position $X$,
		the cops player picks up one cop
		from a vertex $\widehat{x} \in X$ (or a cop that is not placed if $|X| \leq k$).
		This results in the intermediate cop position $W = X \setminus \set{\widehat{x}}$
		(or $W = X$ if the cop was not placed).
		The cops player selects a destination $x \in V(C)$ for this cop.
		\item The robber moves using a $\compEquiv$-compressible $C$-twisting
		fixing all vertices in $\eqclass{\compEquiv}{W}$.
		\item The cop picked up is placed on $x$
		resulting in the next cop position $\widehat{X} = W \cup \set{x}$.
		If the robber is caught (with respect to $\eqclass{\compEquiv}{\widehat{X}}$), then the robber loses.
	\end{enumerate}
	Since this process repeats,
	we can also consider the following as one round in this proof:
	\begin{enumerate}
		\item In the current intermediate cop position $W$,
		the cops player selects a destination $x \in V(C)$ for the next cop to place.
		\item The robber moves using a $\compEquiv$-compressible $C$-twisting
		fixing all vertices in $\eqclass{\compEquiv}{W}$.
		\item The cops player places a cop on $x$ resulting in the cop position
		$X = W \cup \set{x}$.
		If the robber is caught (with respect to $\eqclass{\compEquiv}{X}$), then the robber loses.
		Otherwise, the cops player picks up a cop
		from a vertex $\widehat{x} \in X$ (or a cop that is not placed if $|X| \leq k$).
		This results in the next intermediate cop position $\widehat{W} =  X\setminus \set{\widehat{x}}$
		(or $\widehat{W} = X$ if the cop was not placed).
	\end{enumerate}
	We say that in a cop position $X$ (so in particular in an intermediate cop position)
	the robber is \defining{located in a cop-free row},
	if the robber is located in a row not containing any vertex of $\eqclass{\compEquiv}{X}$.
	We prove by induction on the number of rounds $r \leq \ell$
	that the robber has a strategy such that the
	intermediate cop position $W$ after $r$ rounds satisfies the following two invariants:
	\begin{enumerate}[label=(I.\arabic*)]
		\item \label{inv-no-cop} The robber is located at a cop-free column among the first $k+2$ or the last $k+2$ columns of the cylindrical grid~$C$.
		\item \label{inv-distance}
			If $\eqclass{\compExtEquiv}{W}$ is a toroidal vertical separator, then the following holds:
			\begin{itemize}
			 \item If the robber is located within the first $k+2$ columns,
				then the distance in the toroidal grid~$C^*$ between column $0$ (i.e., the set $I \times \set{0}$) and
				the unique $k$-vertex toroidal vertical separator $S_W \subseteq \eqclass{\compExtEquiv}{W}$ (Lemma~\ref{lem:only-one-separator}) is at least $|J|/3 - 2r$.
			\item If the robber is located within the last $k+2$ columns,
				then the distance in the toroidal grid~$C^*$ between column $|J|-1$ (i.e., the set $I \times \set{|J|-1}$) and
				the unique $k$-vertex toroidal vertical separator $S_W \subseteq \eqclass{\compExtEquiv}{W}$ (Lemma~\ref{lem:only-one-separator}) is at least $|J|/3 - 2r$.
			\end{itemize}
	\end{enumerate}
	Note that the invariant is independent of the representative vertices
	chosen in the cop position
	because all vertices in the first and last $k+2$ columns of $C$ are
	singleton $\compEquiv$-equivalence classes.
	Moreover, $\eqclass{\compExtEquiv}{W}$ is independent of the chosen cop position because $\compEquiv$-equivalent vertices are always $\compExtEquiv$-equivalent.

	In the beginning of the game, the invariant clearly holds because no cop is placed.
	So assume $r < \ell$,
	let~$W$ be the current intermediate cop position,
	and let $x \in V(C)$ be the chosen destination for the next cop to place,
	which will result in the cop position $X \coloneqq W \cup \set{x}$.
	Now it is the robber's turn to move.

	By Invariant~\ref{inv-no-cop},
	the robber is located within the first or last $k+2$ columns of $C$.
	Assume that the robber is located within the first $k+2$ columns
	(the argument for the last $k+2$ columns is analogous).
	We make the following case distinction:

	\begin{enumerate}
		\item \label{itm:case-separator} Assume $\eqclass{\compExtEquiv}{W}$ is a toroidal vertical separator.
		The robber picks a column among the first $k+2$
		that does not contain a vertex of~$X$.
		Such a column exists because~$X$ contains vertices from at most $k+1$ columns.
		By Invariant~\ref{inv-distance},
		every vertical separator has distance at least $k+2$ to the first column.
		Hence, $\eqclass{\compExtEquiv}{W}$ does not separate the first $k+2$ columns and
		so there is a path from the robber to the picked column within the first $k+2$ columns.
		Hence, this path only uses vertices in singleton $\compEquiv$-equivalence classes and thus induces a $\compEquiv$-compressible $C$-twisting avoiding~$\eqclass{\compEquiv}{W}$.
		The robber moves to this column.
		Now a cop is placed on~$x$.
		By construction, the column of the robber is still cop-free
		in the cop position~$X$.
		In particular, the robber does not lose in this round.

		Now a cop is picked up again
		resulting in the intermediate cop
		position $\widehat{W}\subseteq X$.
		Because
		the robber was in a cop-free row in the cop position $X$,
		the robber is still in a cop-free row
		and hence Invariant~\ref{inv-no-cop} holds.

		Let $S_W\subseteq \eqclass{\compExtEquiv}{W}$ be the unique $k$-vertex vertical separator contained in $\eqclass{\compExtEquiv}{W}$ (see Lemma~\ref{lem:only-one-separator}).
		In particular, $W$ is a pseudo-separator.
		Because $k$-vertex pseudo-separators have to contain exactly one vertex per row,
		there is at most one other $k$-vertex pseudo-separator $W' \subseteq X$
		distinct from~$W$
		($W'$ contains~$x$ but not the vertex in the same row in~$W$).
		We make another case distinction:
		\begin{itemize}
			\item Assume the other pseudo-separator~$W'$ exists.
			Then, by Lemma~\ref{lem:pseudo-adjacent-to-sep-is-sep},
			$\eqclass{\compExtEquiv}{W'}$ contains a unique $k$-vertex toroidal vertical separator $S_{W'} \subseteq \eqclass{\compExtEquiv}{W'}$
			differing by at most one vertex from~$S_W$.
			This vertex has distance at most~$2$ to~$S_W$.
			If $\eqclass{\compExtEquiv}{\widehat{W}}$ is a toroidal vertical separator,
			then $\widehat{W} = W'$ or $\widehat{W} = W$ (because~$\widehat{W}$ is in particular a pseudo-separator)
			and the unique $k$-vertex toroidal vertical separator contained in $\eqclass{\compExtEquiv}{\widehat{W}}$ is either~$S_W$ or~$S_{W'}$.
			The distance between $S_W$ and the first column
			is at least $|J|/3 -2(r-1)$ by the inductive hypothesis.
			Because the new vertex in $S_{W'}$ has distance at most $2$ to $S_{W}$, the distance between $S_{W'}$ and the first column
			is at least $|J|/3 -2r$.
			Hence, Invariant~\ref{inv-distance} holds.
			If otherwise $\eqclass{\compExtEquiv}{\widehat{W}}$ is not a toroidal vertical separator, then Invariant~\ref{inv-distance} trivially holds.

			\item Otherwise~$W$ is the only $k$-vertex pseudo-separator contained in~$X$.
			If $\eqclass{\compExtEquiv}{\widehat{W}}$ is a toroidal vertical separator, then $\widehat{W} = W$.
			By the inductive hypothesis,
			$S_{W}$ has distance $|J|/3 -2(r-1) \geq |J|/3 -2r$ to the first column.
			Hence, Invariant~\ref{inv-distance} holds.
			If otherwise~$\eqclass{\compExtEquiv}{\widehat{W}}$ is not a toroidal vertical separator, then Invariant~\ref{inv-distance} trivially holds.
		\end{itemize}

		\item Assume $\eqclass{\compExtEquiv}{W}$ is a pseudo-separator but not a toroidal vertical separator.
		In particular, $\eqclass{\compEquiv}{W}$ does not separate the first $k+2$ columns of $C$.
		The robber moves to a column among the first $k+2$ not containing a vertex of $X$ exactly as in Case~\ref{itm:case-separator}.
		A cop is placed on~$x$.
		The robber is still in a cop-free column
		and does not lose in this round.
		A cop is picked up again resulting in the intermediate cop position $\widehat{W}\subseteq X$.
		Invariant~\ref{inv-no-cop} holds
		because the robber is still in a cop-free column.

		As in Case~\ref{itm:case-separator},
		there can be at most one other $k$-vertex pseudo-separator $W' \subseteq  X$
		distinct from~$W$.
		If $\widehat{W}$ is not a pseudo-separator,
		then in particular $\eqclass{\compExtEquiv}{\widehat{W}}$ is not a toroidal vertical separator and Invariant~\ref{inv-distance} trivially holds.

		Otherwise, $\widehat{W}$ is a pseudo-separator.
		Then, by Lemmas~\ref{lem:only-one-separator} and~\ref{lem:pseudo-adjacent-to-sep-is-sep},
		$\eqclass{\compExtEquiv}{\widehat{W}}$ cannot be a toroidal vertical separator and Invariant~\ref{inv-distance} trivially holds.

		\item Lastly, assume $\eqclass{\compExtEquiv}{W}$ is not a pseudo-separator.
		There can be at most one $k$-vertex pseudo-separator $W' \subseteq X$
		(because $k$-vertex pseudo-separators contain exactly one vertex per row and $W$ is not a pseudo-separator).
		We make another case distinction:
		\begin{enumerate}
			\item Assume this pseudo-separator $W'$ exists
			and that $\eqclass{\compExtEquiv}{W'}$ is a toroidal vertical separator.
			Then $\eqclass{\compExtEquiv}{W'}$ contains a unique $k$-vertex toroidal vertical separator $S_{W'} \subseteq \eqclass{\compExtEquiv}{W'}$
			by Lemma~\ref{lem:only-one-separator}.
			We distinguish two more cases:
			\begin{itemize}
				\item Assume the first column has distance at least $|J|/3$ to~$S_{W'}$ in the toroidal grid~$C^*$.
				The robber moves to a column among the first $k+2$ not containing a vertex of~$X$.
				A cop is placed on~$x$.
				The robber does not lose in this round and Invariant~\ref{inv-no-cop} holds as before.
				After a cop is picked up, let $\widehat{W} \subseteq X$ be the resulting intermediate cop position.

				If $\eqclass{\compExtEquiv}{\widehat{W}}$ is a toroidal vertical separator,
				then $\widehat{W} = W'$ and the first column has distance
				at least $|J|/3$ to $S_{W'}$
				satisfying Invariant~\ref{inv-distance}.
				Otherwise, $\eqclass{\compExtEquiv}{\widehat{W}}$ is not a toroidal vertical separator and Invariant~\ref{inv-distance} trivially holds.

				\item Otherwise, the first column has distance less than $|J|/3$ to~$S_{W'}$ in the toroidal grid~$C^*$.
				By construction and Lemma~\ref{lem:vertical-separator-k-columns},
				the last column of the cylindrical grid~$C$ has distance
				at least $|J| - |J|/3 - 3k -3 \geq |J|/3$ to $S_{W'}$ in the toroidal grid~$C^*$
				because~$S_{W'}$ spans at most~$k$ consecutive columns in~$C^*$ (see Figure~\ref{fig:grids}).
				Recall that~$C$ is a subgraph of~$C^*$ and that $C$ has length~$|J|$ and~$C^*$ has length $|J^*| = 2\cdot |J|$.

				Because~$W$ is not a pseudo-separator,
				there exists an end-to-end twisting avoiding~$W$
				by Lemma~\ref{lem:no-pseudo-sep-twisting}
				that twists two edges~$e_1$ and~$e_2$.
				Assume that~$e_1$ is the edge in the first column of~$C$.
				Because the robber is in a cop-free column and this twisting exists,
				there is a $W$-avoiding path from the robber to~$e_1$
				within the first $k+2$ columns.
				Via this path and the end-to-end twisting,
				the robber moves to the last column in~$C$.
				From there, the robber moves to a column not containing a vertex of $X$ among the last $k+2$ as seen earlier for the first $k+2$ columns.

				A cop is placed on~$x$.
				Again, the robber does not lose and Invariant~\ref{inv-no-cop} holds.
				Next, a cop is picked up resulting in the intermediate cop position $\widehat{W} \subseteq X$.
				If~$\widehat{W}$ is a toroidal vertical separator,
				then $\widehat{W} = W'$, the robber has distance at least $|J|/3$ to~$S_{W'}$, and thus Invariant~\ref{inv-distance} holds.
				Otherwise, $\widehat{W}$ is not a toroidal vertical separator
				and Invariant~\ref{inv-distance} trivially holds.
			\end{itemize}

			\item Otherwise, the robber moves to a column among the first $k+2$ not containing a vertex of~$X$.
			This is done as in Case~\ref{itm:case-separator}.
			A cop is placed on~$x$.
			The robber does not lose in this round and Invariant~\ref{inv-no-cop} holds as before.
			After a cop is picked up, the resulting intermediate cop position $\widehat{W} \subseteq X$ is not a toroidal vertical separator because, for no $k$-vertex subset $W'\subseteq X$, the set $\eqclass{\compExtEquiv}{W'}$ is a toroidal vertical separator.
			Thus, Invariant~\ref{inv-distance} trivially holds. \qedhere
         \end{enumerate}
	\end{enumerate}
\end{proof}

\begin{lemma}
	\label{lem:game-lower-bound}
	Let~$f,g \colon E \to \FF_2$ twist a single edge contained in the first column.
	Spoiler wins the bijective~$(k+1)$-pebble game played on $\compCFI{C,f}{ \compEquiv}$ and $\compCFI{C,g}{ \compEquiv}$
	but not before $\Omega(|J|)$ many rounds.
\end{lemma}

\begin{proof}
	The equivalence relation~$\compEquiv$ is a compression by Lemma~\ref{lem:relation-is-compression}.
	Also, the functions~$f$ and~$g$ are~$\compEquiv$\nobreakdash-compressible because in the first column all vertices are in singleton~$\compEquiv$\nobreakdash-equivalence classes.
	By Lemma \ref{lem:cops-and-robbers-then-bijective}, Spoiler wins the bijective $(k+1)$-pebble game on non-isomorphic CFI graphs over cylindrical grids with $k$ rows.
	So Spoiler also wins on the precompressed CFI graphs $\precompCFI{C,f}{ \compEquiv}$ and $\precompCFI{C,g}{ \compEquiv}$
	and hence, on the compressed CFI graphs $\compCFI{C,f}{ \compEquiv}$ and $\compCFI{C,g}{ \compEquiv}$ by Lemma~\ref{lem:distinguish-CFI-graphs}.

	By Lemma \ref{ref:lem-robber-lower-bound:new:construction},
	the robber wins the $\Omega(|J|)$-round $(k+1)$-Cops and Robber game on~$C$ and~$\compEquiv$ when the robber starts in the first column.
	So Duplicator wins the $\Omega(|J|)$-round bijective $(k+1)$-pebble game played on $\precompCFI{C,f}{ \compEquiv}$ and $\precompCFI{C,g}{ \compEquiv}$
	(Lemma~\ref{lem:compressed-cops-and-robber})
	and on $\compCFI{C,f}{ \compEquiv}$ and $\compCFI{C,g}{ \compEquiv}$
	(Corollary~\ref{cor:bounds-precompressed-compressed}).
\end{proof}

We are now ready to prove the main result of this paper.

\begin{theorem}
	\label{thm:lower-bound-round-number}
	There is a constant $c \in \mathbb{R}_{> 0}$ such that, for all integers $k \geq 2$ and $n \geq c(k \log k)^2$, there are graphs~$G_n,H_n$ such that
	\begin{enumerate}[label = (\alph*)]
		\item $|V(G_n)| = |V(H_n)| \leq c \cdot k^2 \cdot n$,
		\item $G_n \kequivr{k+1}{n^{k/2}} H_n$, and
		\item $G_n \not\kequiv{k+1} H_n$.
	\end{enumerate}
\end{theorem}

\begin{proof}
	The statement is already known for~$k\leq 2$ (see Section~\ref{sec:lin:lower}).
	Let $1 > \varepsilon > 0$ be a sufficiently small constant and $c \in \mathbb{R}_{> 0}$ be sufficiently large (in particularly larger than the constant known for the case $k\leq 2$).
	Also, let $k \geq 2$ and $n \geq c(k \log k)^2$ be arbitrary integers.
	We set $w \coloneqq \frac{2}{\varepsilon} \cdot \lceil\sqrt{n}\,\rceil$.
	Then $w \geq \frac{2}{\varepsilon} \cdot \sqrt{c} \cdot k \cdot \log k$ which implies that Equation \eqref{eq:w-bound} is satisfied (since $c$ is sufficiently large), and there are coprime numbers~$p_0,\ldots,p_{k-1}$ such that~$w/2 < p_i \leq w$ (see Lemma \ref{lem:primes}).

	Let $f,g \colon E \to \FF_2$ twist a single edge contained in the first column of~$C_{I,J}$,
	where~$I$ and~$J$ are defined as before with respect to the numbers $p_0, \dots, p_{k-1}$.
	Consider~$G_n \coloneqq \compCFI{C_{I,J},f}{\compEquiv}$ and~$H_n \coloneqq \compCFI{C_{I,J},g}{\compEquiv}$.
	Then~$|V(G_n)| = |V(H_n)| = \Theta(k^2 \cdot  w^2) = \Theta(k^2 \cdot n)$ by Lemmas~\ref{lem:compressed-cfi-graph-size} and~\ref{lem:grid-width}.
	In particular, $|V(G_n)| = |V(H_n)| \leq c \cdot k^2 \cdot n$ since $c$ is sufficiently large.

	By Lemma \ref{lem:game-lower-bound}, Spoiler wins the bijective $(k+1)$-pebble game played on~$G_n$ and~$H_n$
	but not before $\Omega(|J|)$ many rounds.
	We have
	\[|J| \geq (w/2)^{k} \geq \left(\frac{1}{\varepsilon} \cdot \sqrt{n}\right)^k \geq \frac{1}{\varepsilon} \cdot n^{k/2}.\]
	which implies that
	\[\varepsilon \cdot |J| \geq n^{k/2}.\]
	Since $\varepsilon > 0$ is sufficiently small,
	Spoiler has no winning strategy in the $n^{k/2}$-round game,
	which implies that $G_n \simeq_{k+1}^{n^{k/2}} H_n$.
\end{proof}

\begin{corollary}
	\label{cor:lower-bound-iteration-number}
	There is a constant $c \in \mathbb{R}_{> 0}$ such that, for all integers $k \geq 2$ and $n \geq c(k \log k)^2$, there are graphs~$G_n,H_n$ such that
	\begin{enumerate}[label = (\alph*)]
		\item $|V(G_n)| = |V(H_n)| \leq c \cdot k^2 \cdot n$,
		\item $k$-WL distinguishes $G_n$ and $H_n$, and
		\item $k$-WL requires at least~$n^{k/2}$ rounds to distinguish~$G_n$ and~$H_n$.
	\end{enumerate}
\end{corollary}

\begin{proof}
	This follows from Lemma \ref{lem:k-wl-k-plus-one-bijective-game} and Theorem \ref{thm:lower-bound-round-number}.
\end{proof}

Finally, Theorem \ref{thm:main} follows directly from Corollary \ref{cor:lower-bound-iteration-number}.

\section{Conclusions}

We prove a lower bound of $\Omega(n^{k/2})$ for the iteration number of the $k$-dimensional Weisfeiler-Leman algorithm on graphs.
This is the first improvement over Fürer's linear lower bound from 2001~\cite{Furer01}.
Furthermore, our lower bound even improves the $n^{\Omega(k)}$ lower bound for the iteration number of $k$-WL on $k$-ary relational structures \cite{GroheLN23} (by establishing the fixed constant $k/2$ for the exponent).
Our lower bound is close to the upper bound of $\bigO(n^{k-1}\log n)$, but we leave it as an open problem where the exact bound is situated.

Our main technical contribution is a novel compression technique for CFI graphs.
There is a well-known connection between the CFI graphs and Boolean XOR-formulas (see \cite{BerkholzN23}).
Hence, our construction can also be viewed as a new compression for XOR-formulas.
Furthermore, the WL-algorithm is related to the polynomial calculus, where the dimension $k$ corresponds to the degree of the polynomials in a derivation \cite{BerkholzG15}.
It will be interesting to explore further implications of our construction for proof complexity.
Indeed, after its first publication, our construction has already been used in proof complexity to show treelike-size vs.~width trade-offs for resolution on graph isomorphism formulas~\cite{BerholzLV24} and size vs.~width and depth vs.~width trade-offs for resolution and size vs.~depth trade-offs for the cutting planes proof system~\cite{deRezendeFJNP2024}.

By a known equivalence between the Weisfeiler-Leman algorithm and the finite-variable fragments of first-order logic \cite{CaiFI92} (also see \cite{Kiefer20}), our lower bound also yields an $\Omega(n^{(k-1)/2})$-lower bound on the quantifier depth of sentences in the $k$-variable fragment of first-order logic (with or without counting) needed to distinguish non-isomorphic graphs of order $n$.

Finally, we remark that $(3k)$-WL can distinguish between the constructed instances using only $\bigO(\log n)$ many rounds (for every fixed $k \geq 2$).
Indeed, using the same ``binary search'' strategy already employed in \cite{Furer01}, $(3k)$-WL only requires $\bigO(\log n)$ rounds to distinguish between $\CFI{C_{I,J},f}$ and $\CFI{C_{I,J},g}$ (assuming $\sum f \neq \sum g$).
Using Lemma \ref{lem:distinguish-CFI-graphs}, the same bounds can be achieved for the compressed versions $\compCFI{C_{I,J},f}{ \compEquiv}$ and $\compCFI{C_{I,J},g}{ \compEquiv}$.
It is an interesting open question whether our construction can be modified so that the iteration number remains large even if the dimension increases.

\renewcommand{\path}{\pathbibtex}
\bibliographystyle{plainurl}
\bibliography{lowerboundwl}

\end{document}